\documentclass[11pt,reqno,a4paper]{amsart}
\pdfoutput=1
\usepackage{amsmath,amssymb,shuffle,enumerate}
\usepackage{microtype}
\usepackage{mathtools}
\usepackage{xcolor}
\usepackage{etoolbox}
\usepackage[hidelinks=true]{hyperref}
\providecommand{\subparagraph}{\paragraph}

\let\noAP\AP
\newif\ifknowledge
\knowledgetrue
\ifknowledge
  \usepackage[final=true]{knowledge}

\knowledgeconfigure{notion}
\knowledgeconfigure{quotation}
\knowledgeconfigure{fix hyperref twocolumn}
\knowledgedirective{normal}{}
\knowledgedirective{math command}{notion}

\knowledge {\command energy parity game}{math command}
\knowledge {\command extended energy game}{math command}
\knowledge {\command bounding game}{math command}
\knowledge {\command perfect half space game}{math command}
\knowledge {\command lexicographic energy game}{math command}
\knowledge {\command mean payoff}{math command}
\knowledge {\command lexicographic ordering}{math command}
\knowledge {\command prefix ordering}{math command}
\knowledge {\command least common prefix}{math command}
\knowledge {\command least common prefix path}{math command}
\knowledge {\command first cycle game}{math command}
\knowledge {\symbol ddagger}{math command}
\knowledge {\symbol dagger}{math command}
\knowledge {\|}{math command}

\knowledge {lexicographic ordering}{notion}
\knowledge {prefix ordering}{notion}

\knowledge {multi-dimensional game graph}
   [multi-dimensional game graphs|Multi-Weighted Game Graphs]
   {notion}
\knowledge {multi-dimensional weighted game graph}{}
\knowledge {$d$-dimensional game graph}{link=multi-dimensional game graph}
\knowledge {$d$-dimensional game graphs}{link=multi-dimensional game graph}
\knowledge {multi-weighted game graph}[multi-weighted game graphs]{link=multi-dimensional game graph}
\knowledge {$d$-dimensional multi-weighted game graph}
   {link=multi-dimensional game graph}
\knowledge {$2d$-dimensional game graph}
  {link=multi-dimensional game graph}
\knowledge {weighted game graph}{notion}
\knowledge {$1$-dimensional weighted game graph}{link=multi-dimensional game graph}

\knowledge {extended multi-dimensional weighted game graph}{notion}
\knowledge {extended (finite) multi-dimensional energy game graph}{synonym}
\knowledge {extended (finite) multi-dimensional weighted game graph}{synonym}

\knowledge {weight}[weights|Weights]{notion}
\knowledge {multi-weight}[multi-weights|Multi-weights]{notion}

\knowledge{perfect half space}
   [perfect half spaces|Perfect half spaces|Perfect Half Spaces]
   {notion}
\knowledge {\kl {perfect half space}}{link=perfect half space}
\knowledge {current perfect half space}{link=perfect half space}
\knowledge {partially perfect half space}[partially perfect half spaces|Partially perfect half spaces]{notion}
\knowledge {$d$-dimensional perfect half spaces}{link=perfect half space}

\knowledge {perfect half space game}[perfect half space games|Perfect half space games]{notion}
\knowledge {\kl {perfect half space games}}{link=perfect half space games}
\knowledge {Perfect Half Space Games}{link=perfect half space games}

\knowledge {lexicographic energy game}
 [lexicographic energy games|Lexicographic energy games]{notion}
\knowledge {lexicographic energy ones}{synonym}
\knowledge {Lexicographic Energy Games}{synonym}

\knowledge {bounding game}[bounding games]{notion}
\knowledge {\kl {bounding games}}{link=bounding game}

\knowledge {mean-payoff game}[mean-payoff games]{notion}

\knowledge {multi-dimensional energy game}[multi-dimensional energy games]{normal}
\knowledge {$d$-dimensional energy game}{link=multi-dimensional energy games}

\knowledge {extended multi-dimensional energy game}
    [extended multi-dimensional energy games]{notion}
\knowledge {`extended' multi-dimensional energy games}
    {link=extended multi-dimensional energy games}
\knowledge {Extended Multi-Dimensional Energy Games}
    {link=extended multi-dimensional energy games}
\knowledge {Extended multi-dimensional energy games}
    {link=extended multi-dimensional energy games}

\knowledge {multi-dimensional energy parity game}
  [multi-dimensional energy parity games]{notion}
\knowledge {Multi-Dimensional Energy Parity Games}{link=multi-dimensional energy parity game}
\knowledge {priority function}{notion}
\knowledge {energy objective}{notion}
\knowledge {parity objective}{notion}
\knowledge {parity function}{notion}

\knowledge{proof:cl-gphs}{notion}

\knowledge {first-cycle game}{notion}

\knowledge {initial credit}{notion}

\knowledge{strategy}[strategies|Strategies]{normal}
\knowledge {positional strategy}{notion}
\knowledge {positional strategies}{synonym}

\knowledge {positional optimal strategy}[positional optimal strategies]{link=positional strategy}
\knowledge {winning positional strategy}{link=positional strategy}
\knowledge {winning positional strategy for Player~$2$}{link=positional strategy}
\knowledge {positional winning strategies}{link=positional strategy}

\knowledge {cycle decomposition}{notion}

\knowledge {positional determinacy}{notion}
\knowledge {positionally determined}{synonym}

\knowledge {leading dimension}{notion}

\knowledge {flag vector}[Flag Vectors|flag vectors]{notion}

\knowledge {interleaving}[interleavings|Interleavings]{notion}
\knowledge {interleaved}{synonym}

\knowledge {oblivious}{notion}
\knowledge {perfect half space oblivious}{synonym}
\knowledge {oblivious strategy}{synonym}

\knowledge {oblivious at}{notion}
\knowledge {oblivious at all vertices}{synonym}
\knowledge {oblivious at $v$}{synonym}
\knowledge {perfect half space oblivious at}{synonym}
\knowledge {(perfect half space) oblivious at~$v$}{synonym}

\knowledge {key remark}{notion}

\knowledge {good}{notion}
\knowledge {bad}{synonym}
\knowledge {good perfect half space}{synonym}

\knowledge {elementary path}{notion}
\knowledge {$(\tau ,v)$-elementary path}{synonym}
\knowledge {$(\tau _\HS ,v)$-elementary path}{synonym}
\knowledge {$(\tau ',v)$-elementary path}{synonym}

\knowledge {$V_1$ cycle decomposition}{notion}
\knowledge {recurring}{notion}
\knowledge {downward path}[downward paths]{notion}

\knowledge {given initial credit problem}
  [given initial credit problems|given initial credit|given initial credits]{notion}
\knowledge {Given Initial Credit}{synonym}

\knowledge {arbitrary initial credit problem}
  [arbitraty initial credit problems|arbitrary initial credit|arbitrary initial credits]{notion}
\knowledge {Arbitrary Initial Credit}{synonym}

\knowledge {multi-dimensional mean-payoff games}{normal}
\knowledge {energy games}{normal}
\knowledge {pwin}{normal}
\knowledge {winning strategy}{normal}

\knowledge {norm PHS}{notion}

\knowledge {representation PHS}{notion}
\knowledge {partially perfect half space representation}{synonym}

\knowledge {win PHS}{notion}

\knowledge {winning play PHS}{notion}

\else
  \makeatletter
  \def\intro{\@ifstar\klintro\klintro}
  \newcommand\klintro[2][]{\ifmmode\mathit{#2}\else\emph{#2}\fi}
  \newcommand\reintro[2][]{#2}
  \newcommand\kl[2][]{#2}
  \newcommand\AP\relax
  \makeatother
\fi
\let\knowledgeAP\AP
\let\AP\noAP

\let\oldsec\S
\usepackage[small,basic]{complexity}
\def\isnum#1{%
  \if!\ifnum9<1#1!\else_\fi
    \ComplexityFont{#1\text{-}}\else#1\fi}
\renewcommand{\EXP}[1][]{{\isnum{#1}\ComplexityFont{EXPTIME}}}

\newclass{\TOWER}{TOWER}
\renewcommand{\S}{\oldsec}
\newrobustcmd{\qedmath}{\tag*{\qed}}
\newrobustcmd{\tuple}[1]{(#1)}
\newrobustcmd{\lossy}[1]{\mathsf{lossy}\!\left(#1\right)}
\newrobustcmd{\capped}[2][i]{\mathsf{capped}^{#1}_{\vec
    b}\!\left(#2\right)}
\newrobustcmd{\avg}[1]{\mathsf{avg}(#1)}
\newrobustcmd{\tsum}{\textstyle{\sum}}
\newrobustcmd{\vsp}[1]{\mathsf{span}(#1)}
\newrobustcmd{\cone}[1]{\mathsf{cone}(#1)}
\newrobustcmd{\ohsp}[2]{\vsp{#1}^{\rightarrow}_{#2}}
\newrobustcmd\?[1]{\ensuremath{\mathcal{#1}}}
\newrobustcmd\+[1]{\ensuremath{\mathbb{#1}}}
\let\normsymbol\|
\renewrobustcmd\|{\kl[\|]{\normsymbol}}

\newrobustcmd{\eqby}[1]{\stackrel{\text{{\tiny{#1}}}}{=}}
\newrobustcmd{\eqdef}{\eqby{def}}
\renewcommand{\vec}[1]{\mathbf{#1}}

\newrobustcmd\phs{\mathop{\kl[\command least common prefix]{\mathrm{lcp}}}}
\newrobustcmd\phspath{\mathop{\kl[\command least common prefix path]{\mathrm{lcp}}}}

\newrobustcmd\weight{\vec w}

\newrobustcmd\leqpref{\mathrel{\kl[\command prefix ordering]{\leq^{\text{pref}}}}}
\newrobustcmd\lpref{\mathrel{\kl[\command prefix ordering]{<^{\text{pref}}}}}
\newrobustcmd\gamefont{\mathtt}
\newrobustcmd\PHSgame{\kl[\command perfect half space game]{\gamefont{PHS}}}
\newrobustcmd\LexEngame{\kl[\command lexicographic energy game]{\gamefont{LexEn}}}
\newrobustcmd\Bndgame{\kl[\command bounding game]{\gamefont{Bnd}}}
\newrobustcmd\ExtEngame{\kl[\command extended energy game]{\gamefont{ExtEn}}}
\newrobustcmd\EnPgame{\kl[\command energy parity game]{\gamefont{EnPar}}}
\newrobustcmd\FCgame{\kl[\command first cycle game]{\gamefont{FC}}}%
\newrobustcmd\MPgame{\kl[\command mean payoff]{\gamefont{MP}}}

\let\symbolprec\prec
\newrobustcmd\newprec{\mathrel{\kl[\command lexicographic ordering]{\symbolprec}}}
\let\prec\newprec

\let\symbolpreceq\preceq
\newrobustcmd\newpreceq{\mathrel{\kl[\command lexicographic ordering]{\symbolpreceq}}}
\let\preceq\newpreceq

\let\symbolsucc\succ
\newrobustcmd\newsucc{\mathrel{\kl[\command lexicographic ordering]{\symbolsucc}}}
\let\succ\newsucc

\let\symbolsucceq\succeq
\newrobustcmd\newsucceq{\mathrel{\kl[\command lexicographic ordering]{\symbolsucceq}}}
\let\succeq\newsucceq

\let\symboldagger\dagger
\newrobustcmd\newdagger{{\kl[\symbol dagger]{\symboldagger}}}
\let\dagger\newdagger

\let\symbolddagger\ddagger
\newrobustcmd\newddagger{{\kl[\symbol ddagger]{\symbolddagger}}}
\let\ddagger\newddagger

\newif\ifshort
\shorttrue
\makeatletter
\providecommand{\@IEEEsectpunct}{}
\providecommand{\nopunct}{\def\@IEEEsectpunct{\ }}
\makeatother
\usepackage[numbers]{natbib}\renewcommand{\cite}{\citep}
\def\bibfont{\footnotesize}
\makeatletter
\def\NAT@spacechar{~}%
\newcounter{NAT@nlabs}
\makeatother

\usepackage[disable,backgroundcolor=orange!50,textsize=scriptsize]{todonotes}
\presetkeys{todonotes}{inline}{}
\usetikzlibrary{positioning}
\usetikzlibrary{shapes.geometric}
\usetikzlibrary{automata}
\usetikzlibrary{arrows}
\tikzstyle{triangle}=[rounded corners=.5,regular polygon,regular polygon
  sides=3,minimum size=.5cm,draw=gray!90,inner sep=1pt,fill=gray!20,very thick]
\tikzstyle{square}=[rounded corners=.5,regular polygon,regular polygon
  sides=4,minimum size=.5cm,draw=gray!90,inner sep=1pt,fill=gray!20,very thick]
\tikzstyle{every node}=[font=\small]
\tikzstyle{every edge}=[draw,>=stealth',shorten >=1pt,semithick]
\tikzstyle{accepting}=[accepting by arrow]
\tikzstyle{initial}=[initial by arrow,initial text=]
\tikzstyle{state}=[draw=black!70,very thick,fill=black!20,circle,inner
  sep=2pt]
\usepackage{amsthm,thmtools,thm-restate}
\declaretheorem[numberwithin=section]{theorem}
\declaretheorem[sibling=theorem]{proposition}
\declaretheorem[sibling=theorem]{lemma}
\declaretheorem[sibling=theorem]{corollary}
\declaretheorem[sibling=theorem]{definition}
\declaretheorem[sibling=theorem]{fact}

\declaretheorem[sibling=theorem,style=remark,qed=\qedsymbol]{example}
\declaretheorem[sibling=theorem]{claim}
\usepackage[capitalise,nameinlink]{cleveref} 
\crefname{claim}{Claim}{Claims}
\Crefname{claim}{Claim}{Claims}
\crefname{fact}{Fact}{Facts}
\Crefname{fact}{Fact}{Facts}
\crefname{figure}{Figure}{Figures}
\Crefname{figure}{Figure}{Figures}

\let\AP\knowledgeAP
\begin{document}
\renewcommand{\sectionautorefname}{Section}
\renewcommand{\subsectionautorefname}{Section}
\renewcommand{\subsubsectionautorefname}[1]{\S}
\title{Perfect Half Space Games}
\author[Th.~Colcombet]{Thomas Colcombet}
\address{IRIF, CNRS \& Universit\'e Paris-Diderot, France\vspace*{-1em}}
\author[M.~Jurdzi\'nski]{Marcin Jurdzi\'nski}%
\author[R.~Lazi\'c]{Ranko Lazi\'c} \address{DIMAP, Department of
  Computer Science, University of Warwick, UK}
\author[S.~Schmitz]{Sylvain Schmitz} \address{LSV, ENS Paris-Saclay \&
CNRS \& INRIA, Universit\'e Paris-Saclay, France}

\begin{abstract}
We introduce \kl{perfect half space games}, in which the goal of
Player~2 is to make the sums of encountered multi-dimensional weights
diverge in a direction which is consistent with a chosen sequence of
perfect half spaces (chosen dynamically by Player~2).  We establish
that the \kl{bounding games} of Jurdzi\'nski et al.~(ICALP~2015) can
be reduced to \kl{perfect half space games}, which in turn can be
translated to the \kl{lexicographic energy games} of Colcombet and
Niwi\'nski, and are \kl{positionally determined} in a strong sense
(Player~2 can play without knowing the current perfect half space).
We finally show how \kl{perfect half space games} and \kl{bounding
games} can be employed to solve \kl{multi-dimensional energy parity
games} in pseudo-polynomial time when both the numbers of energy
dimensions and of priorities are fixed, regardless of whether
the \kl{initial credit} is given as part of the input or existentially
quantified.  This also yields an optimal {\footnotesize\textsf{2-EXPTIME}}
complexity with given \kl{initial credit}, where the best known upper
bound was non-elementary.

\end{abstract}
\maketitle

\section{Introduction}

A \intro{$d$-dimensional energy game}~\citep{brazdil10,velner15} sees
two players compete in a finite game graph, whose edges are decorated
with vectors of weights in $\+Z^d$.  The $d$ weights represent various
discrete resources that can be consumed or replenished by the actions
of the game.  The objective of Player~1, given an \kl{initial credit} in
$\+N^d$, is to play indefinitely without depleting any of the
resources---more precisely to keep the current sum of encountered
weights plus \kl{initial credit} non-negative in every dimension---while
Player~2 attempts to foil this.
The primary motivation for these games is controller synthesis for
resource-sensitive reactive systems, where they are also closely
related to \kl{multi-dimensional mean-payoff games}---and actually
equivalent if finite-memory strategies are sought for the
latter~\citep[Lemma~6]{velner15}.  But they appear in diverse settings:
for example, in process algebra, they are equivalent to the simulation
problem between a finite state system and a Petri net or a basic
parallel process~\citep[Propositions~6.2 and~6.4]{courtois14}; in artificial
intelligence, they allow to solve the model-checking problem for the
resource-bounded logic RB$\pm$ATL~\citep{alechina17,alechina16}.

The algorithmic issues surrounding \kl{multi-dimensional energy games} have
come under considerable scrutiny.  Deciding whether there exists an
\kl{initial credit} that would allow Player~1 to win
is \coNP-complete~\citep[Theorem~3]{velner15}, while the complexity when
the \kl{initial credit} is given as part of the input
becomes \EXP[2]-complete~\citep{courtois14,jurdzinski15}.
Finally, both decision problems are in pseudo-polynomial time when
$d$ is fixed~\citep{jurdzinski15}.

\subsubsection*{Open Questions}
However, these recent advances do not settle the case of
\kl{multi-dimensional energy parity games}~\citep{chatterjee14},
where Player~1 must ensure that, in addition to the quantitative
energy objective (specifying resource consumption and replenishment),
she also complies with a qualitative $\omega$-regular objective in the
form of a parity condition (specifying functional requirements).
These games with arbitrary \kl{initial credit} are
still \coNP-complete as a consequence
of~\citep[Lemma~4]{chatterjee14}.  With given \kl{initial credit},
they were first proven decidable by \citet*{abdulla13}, and used to
decide both the model-checking problem for a suitable fragment of the
$\mu$-calculus against Petri net executions and the \emph{weak}
simulation problem between a finite state system and a Petri net; they
also allow to decide the model-checking problem for the resource logic
RB$\pm$ATL$^*$~\citep{alechina16}.  As shown by \citet{jancar15},
$d$-dimensional energy games using $2p$ priorities can be reduced
to \kl{`extended' multi-dimensional energy games} of dimension
$d'\eqdef d+p$, with complexity upper bounds shown earlier
by \citet*{brazdil10} to be in
\EXP[$(d'-1)$-] when $d'\geq 2$ is fixed, and in \TOWER\ when $d'$ is
part of the input, %
leaving a substantial complexity gap with the
\EXP[2]-hardness shown in~\citep{courtois14}.

\subsubsection*{Contributions}
We introduce in \cref{sec-phs} \emph{\kl{perfect half space games}}, both
\begin{itemize}
\item as intermediate objects in a chain of reductions from
multi-dimensional energy parity games to mean-payoff games
(see \cref{fig:reds}), allowing us to derive new tight complexity
upper bounds based on recent advances by \citet*{comin16} on the
complexity of mean-payoff games, and
\item as a means to gain a deeper understanding of how winning
strategies in energy games are structured.
\end{itemize}

More precisely, in perfect half-space games, positions are pairs: a
vertex from a \kl{$d$-dimensional game graph} as above, together with
a $d$-dimensional \emph{\kl{perfect half space}}.  The latter is a
maximal salient blunt cone in $\+Q^d$: a union of open half spaces of
dimensions $d$, $d - 1$, \ldots, 1, where each is contained in the
boundary of the previous one.  In these games, Player~1 may not change
the \kl{current perfect half space}, but Player~2 may change it
arbitrarily at any move.  However, the goal of Player~2 is to make the
sums of encountered weights diverge in a direction which is consistent
with the chosen \kl{perfect half spaces}; thus the greater the
dimension of the component open half spaces that Player~2 varies
infinitely often, the harder it is for him to win.  For example, with
$d = 2$, if Player~2 eventually settles on the perfect half space that
consists of the half plane $x < 0$ and the half line $x = 0 \,\wedge\,
y < 0$, then he wins provided the sequence of total weights is such
that either their $x$-coordinates diverge to $-\infty$, or their
$x$-coordinates do not diverge to $+\infty$ and their $y$-coordinates
diverge to $-\infty$; if however Player~2 switches between the two
half lines of $x = 0$ infinitely often, then he can only win in the
former manner.

Firstly, we show that perfect half space games can be easily
translated to the \reintro{lexicographic energy games} of 
\citet*{colcombet17}.  The translation amounts to normalising the edge
weights with respect to the current perfect half spaces, and inserting
another $d$ dimensions in which we encode appropriate penalties for
Player~2 that are imposed whenever he changes the perfect half space
(cf.\ \cref{sub-tr-lmp}).  We deduce that \kl{perfect half space games} are
\kl{positionally determined}, and moreover that Player~2 has winning
strategies that are \kl{oblivious} to the current perfect half space.
Along the way, we provide in \cref{sub-lmp} a proof of the \kl{positional
determinacy} of \kl{lexicographic energy games}, along with pseudo-polynomial
complexity upper bounds for their decision problem when $d$ is fixed,
based on the recent results of \citet{comin16} for \kl{mean-payoff games}.

Secondly, we establish that \kl{perfect half space games} capture
\reintro{bounding games} (cf.\ \cref{sec-bnd}).  The latter were central
to obtaining the tight complexity upper bounds for \kl{multi-dimensional
energy games}~\citep{jurdzinski15}.  They are played purely on the
\kl{$d$-dimensional game graphs} and have a simple winning
condition: the goal of Player~1 is to keep the total absolute value of weights bounded
(i.e., contained in some $d$-dimensional hypercube).  One reading of
this reduction is that whenever Player~2 has a winning
strategy in a bounding game, he has one that `announces' at every move
some \kl{perfect half space} and succeeds in forcing the total weights to
be unbounded in a direction consistent with the infinite sequence of
his announcements.  The proof  is difficult, and
relies on a construction from the previous
paper~\citep{jurdzinski15} of a winning strategy for Player~1 in the
\kl{bounding game} given her winning strategy in a \kl{first-cycle game}
featuring \kl{perfect half spaces}.  Composing this with our complexity
bounds for \kl{lexicographic energy games} gives us a new approach to
solving \kl{bounding games}, improving the time complexity from the
previously best
$(|V| \cdot \|E\|)^{O(d^4)}$~\citep[Corollary~3.2]{jurdzinski15} to
$(|V| \cdot \|E\|)^{O(d^3)}$, where $V$ is the set of vertices and
$\|E\|$ the maximal absolute value over the weights in the input
\kl{multi-dimensional game graph} (cf.\ \cref{cor-bnd}).

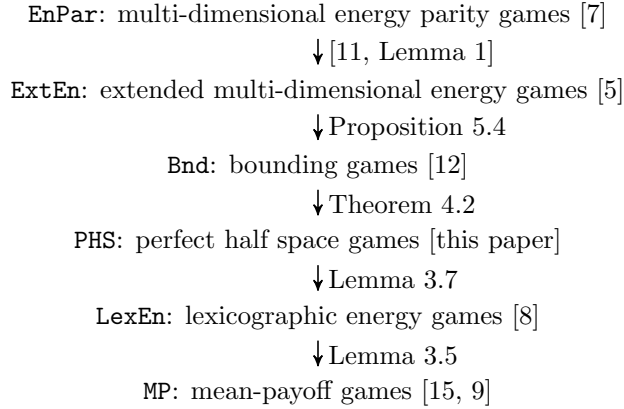
\begin{figure}[tbp]
\centering
\begin{tikzpicture}[on grid,auto]
  \node(epg){$\EnPgame$: "multi-dimensional energy parity games"~\citep{chatterjee14}};
  \node[below=of epg](eeg)
    {$\ExtEngame$: "extended multi-dimensional energy games"~\citep{brazdil10}};
  \node[below=of eeg](bg)
    {$\Bndgame$: "bounding games"~\citep{jurdzinski15}};
  \node[below=of bg](phs)
    {$\PHSgame$: "perfect half space games"~[this paper]};%
  \node[below=of phs](leg)
    {$\LexEngame$: "lexicographic energy games"~\citep{colcombet17}};
  \node[below=of leg](mpg)
    {$\MPgame$: "mean-payoff games"~\citep{zwick96,comin16}};
  \path[->]
    (epg) edge node{\citep[Lemma~1]{jancar15}} (eeg)
    (eeg) edge node{\cref{pr:aic}} (bg)
    (bg) edge node{\cref{th:eq.bounding}} (phs)
    (phs) edge node{\cref{th:tr.lmp.1}} (leg)
    (leg) edge node{\cref{lem:lmpg-positional}} (mpg);
\end{tikzpicture}
\caption{\label{fig:reds}The reductions between the various games in
  this paper.}
\end{figure}

Thirdly, building on \citeauthor{jancar15}'s reduction, we show how
\kl{multi-dimensional energy parity games} can be solved by means of
bounding games (cf.\ \cref{sec-mep}).  %
For the \kl{given initial credit} problem, we
obtain \EXP[2]-completeness, closing the aforementioned complexity
gap.  When the dimension $d$ and the number of priorities $2p$ are
fixed, %
we obtain that, for both
arbitrary and \kl{given initial credits}, the winner is decidable in
pseudo-polynomial time.  With \kl{arbitrary initial credit}, our new
bound $(|V|\cdot\|E\|)^{O((d+p)^3 \log(d+p))}$ improves when $p=0$
over the previously best $(|V|
\cdot\|E\|)^{O(d^4)}$~\citep[Theorem~3.3]{jurdzinski15}.

\subsubsection*{Structure of the Paper}  The chain of reductions we
use in this paper is depicted in \cref{fig:reds}, and we shall
essentially work our way up through it.  In \cref{sec-phs} we
introduce \kl{multi-dimensional game graphs} and \kl{perfect half
space games}. In \cref{sec-lmp} we show how to employ \kl{lexicographic
energy games} %
for solving \kl{perfect half space
games}.  We apply these results to \kl{bounding games}
in \cref{sec-bnd} and \kl{multi-dimensional energy parity games}
in \cref{sec-mep}, before concluding.

\section{Perfect Half Space Games}\label{sec-phs}

\subsection{\kl{Multi-Weighted Game Graphs}}

We consider \intro{multi-dimensional game graphs} whose edges are labelled by \intro{multi-weights}, which are
vectors of integers.  They are tuples of the form $\tuple{V, E, d}$,
where $d$ is the dimension in $\+N_{>0}$, $V\eqdef V_1\uplus V_2$ is a
finite set of vertices partitioned into Player~$1$ vertices and
Player~$2$ vertices, and $E$ is a finite set of edges included in
$V \times\+Z^d\times V$, such that every vertex has at least one
outgoing edge.  We may write just `\intro{weight}' instead of `\kl{multi-weight}' 
when there is no risk of confusion, and also $v\xrightarrow{\vec w}v'$ 
to denote an edge $(v,\vec w,v')$. Given a path~$P$ in the game,
we denote by $\weight(P)$ the sum of the \kl{weights} encountered.

\AP 
For a vector $\vec w$ in $\+Z^d$, we let $\intro*\|\vec w\reintro*\|\eqdef\max_{1\leq
i\leq d}|\vec w(i)|$ denote its infinity norm; we define the norm
$\|E\|\eqdef\max_{v\xrightarrow{\vec w}v'\in E}\|\vec w\|$ as the
maximum of the norms of edge weights.  We assume
all our integers to be encoded in binary, hence $\|E\|$ might be
exponential in the size of the multi-weighted game graph.

Without loss of generality, we assume that the players
strictly alternate ($v\xrightarrow{\vec w}v'$ in $E$ implies $v$ in
$V_i$ and $v'$ in $V_{3-i}$ for some $i$ in $\{1,2\}$), the weight of
every edge is determined by its vertices ($v\xrightarrow{\vec w}v'$
and $v\xrightarrow{\vec w'}v'$ in $E$ implies $\vec w=\vec w'$), and
not all weights are zero ($\|E\|>0$).

\begin{example}\label{ex:wgg}
  \Cref{fig-mean} shows on its left-hand-side an example of a
  2-dimensional weighted game graph.  Throughout this paper,
  Player~$1$ vertices are depicted as triangles and Player~$2$ vertices as
  squares.
\end{example}%

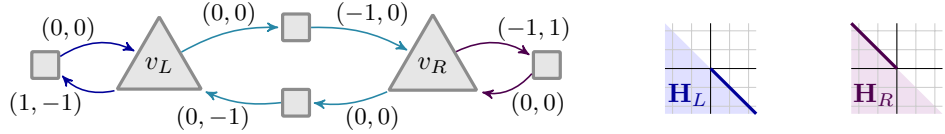
\begin{figure*}[tbp]
  \centering
  \begin{tikzpicture}[auto,on grid]
    \node[triangle](L){$v_L$};
    \node[triangle, right=3.6cm of L](R){$v_R$};
    \node[square,above right=.5 and 1.8 of L](A){};
    \node[square,below right=.5 and 1.8 of L](B){};
    \node[square, left=1.5 of L](L1){};
    \node[square, right=1.5 of R](R1){};
    \path[->,every node/.style={font=\footnotesize,inner sep=.5pt}]
      (L) edge[bend left=15,draw=black!40!cyan,pos=0.8] node{$\tuple{0,0}$} (A)
      (A) edge[bend left=15,draw=black!40!cyan,pos=0.2] node{$\tuple{-1,0}$} (R)
      (R) edge[bend left=15,draw=black!40!cyan,pos=0.6] node{$\tuple{0,0}$} (B)
      (B) edge[bend left=15,draw=black!40!cyan,pos=0.35] node{$\tuple{0,-1}$} (L)
      (L) edge[bend left,draw=black!40!blue] node{$\tuple{1,-1}$} (L1)
      (L1) edge[bend left,draw=black!40!blue] node{$\tuple{0,0}$} (L)
      (R) edge[bend left,draw=black!40!violet] node{$\tuple{-1,1}$}(R1)
      (R1) edge[bend left,draw=black!40!violet] node{$\tuple{0,0}$}(R)
  ;
  \end{tikzpicture}\hspace*{1cm}
  \begin{tikzpicture}
    \fill[blue!12] (-.6,.6) -- (.6,-.6) -- (-.6,-.6);
    \draw[step=.25cm,gray!40,very thin] (-.6,-.6) grid (.6,.6);
    \draw[very thin] (-.6,0) -- (.6,0);
    \draw[very thin] (0,-.6) -- (0,.6);
    \draw[very thick,black!40!blue] (0,0) -- (.6cm,-.6cm);
    \node[color=black!40!blue] at (-.3,-.35) {$\vec H_L$};
    \node at (0,-.7) {};
  \end{tikzpicture}\hspace*{1cm}
  \begin{tikzpicture}
    \fill[violet!12] (-.6,.6) -- (.6,-.6) -- (-.6,-.6);
    \draw[step=.25cm,gray!40,very thin] (-.6,-.6) grid (.6,.6);
    \draw[very thin] (-.6,0) -- (.6,0);
    \draw[very thin] (0,-.6) -- (0,.6);
    \draw[very thick,black!40!violet] (0,0) -- (-.6cm,.6cm);
    \node[color=black!40!violet] at (-.3,-.35) {$\vec H_R$};
    \node at (0,-.7) {};
  \end{tikzpicture}%
  \caption{A 2-dimensional game graph $(V,E,2)$ and two perfect half spaces.\label{fig-mean}}%
\end{figure*}

\subsection{\kl{Perfect Half Spaces}}

\AP
We ""represent@representation PHS"" \kl{partially perfect half spaces} by tuples $\vec{H} =
\tuple{\vec{h}_1, \ldots, \vec{h}_k}$ of $k\leq d$ mutually normal
nonzero $d$-dimensional integer vectors, which are 
normal to the represented half spaces.  For this, 
let $\intro*\prec$ denote the (strict) \emph{\intro{lexicographic ordering}}%
, and for any $d$-dimensional vector
$\vec{a}$, let $\vec{a} \cdot \vec{H}$ denote the pointwise
dot-product $\tuple{\vec{a} \cdot \vec{h}_1, \ldots, \vec{a} \cdot
  \vec{h}_k}$.  The \intro{partially perfect half space} denoted by
$\vec{H}$ is then $\{\vec{a} \in \+Q^d \,:\, \vec{a} \cdot \vec{H}
\prec \vec{0}\}$.

\AP
Let $|\vec H|\eqdef k$.
When $|\vec H| = d$, the "representation@representation PHS"
is a (full) \intro{perfect half space}; 
when $|\vec H| = 0$, it is the empty set since 
there is only one $0$-dimensional vector and the ordering $\prec$ is strict.

We define the ""norm@norm PHS"" $\|\vec{H}\|$ as the maximum of 
$\|\vec{h}_1\|$, \ldots,~$\|\vec{h}_k\|$.

\begin{example}
\label{ex:phs}
The two \kl{perfect half spaces} of interest on the right-hand side of
\cref{fig-mean} are $\{(x,y)\,:\, x+y<0\}\cup\{(x,y)\,:\,x+y=0\wedge
x>0\}$ denoted by $\vec H_L\eqdef\tuple{\tuple{1,1}, \tuple{-1, 1}}$,
and $\{(x,y)\,:\, x+y<0\}\cup\{(x,y)\,:\,x+y=0\wedge x<0\}$ denoted by
$\vec H_R\eqdef\tuple{\tuple{1,1}, \tuple{1, -1}}$.  They have the
half-plane $\{(x,y)\,:\, x+y<0\}$ with normal vector $(1,1)$ in
common, but differ on which half-line of its boundary
$\{(x,y)\,:\,x+y=0\}$ they contain.
\end{example}

\AP
We shall reason sometimes directly on the "representations@representation PHS"
of \kl{partially
perfect half spaces} through the \intro{prefix ordering}.  We write
$\vec H\mathbin{\intro*\leqpref}\vec H'$ when $\vec H$ is a prefix of $\vec
H'$, and 
$\intro*\phs_i\vec H_i$
for the longest common
prefix of a finite or infinite set of \kl{partially perfect half spaces}
$\vec H_1,\vec H_2,\dots$.  Observe that, if $\vec a\cdot\vec
H'\prec\vec 0$ and $\vec H\leqpref\vec H'$,
then $\vec a\cdot\vec H\preceq\vec 0$.

\subsection{\kl{Perfect Half Space Games}}

We write $\tuple{\widehat{V}, \widehat{E}, d}$ for 
the "weighted game graph" obtained from $\tuple{V, E, d}$ by pairing 
vertices in $V$ with \kl{perfect half spaces} of appropriately bounded norms, 
which may be changed only by Player~$2$:
\begin{itemize}
\item for both $i \in \{1, 2\}$, $\widehat{V}_i\eqdef V_i \times \?H$
  where $\?H$ is the set of all \kl{perfect half spaces} of norm at
  most $|V| \cdot \|E\|$;
\item
$\widehat{E}$ is the set of all
$\tuple{v, \vec{H}} \xrightarrow{\vec{w}} \tuple{v', \vec{H}'}$
such that $v \xrightarrow {\vec{w}} v'$ is in $E$
and if $v \in V_1$ then $\vec{H} = \vec{H}'$.
\end{itemize}

\AP
Let $\intro*\PHSgame\tuple{\widehat{V}, \widehat{E}, d}$ denote
the \intro{perfect half space game} in which the goal of Player~$2$ is
for the total weight to diverge in a direction consistent with the
chosen perfect half spaces:
\begin{definition}[Winning Condition for Perfect Half-Space Games]\label{eq-phs}
  An infinite play
  $\tuple{v_0, \vec{H}_0} \xrightarrow{\vec{w}_1} \tuple{v_1,
    \vec{H}_1} \xrightarrow{\vec{w}_2} \tuple{v_2, \vec{H}_2} \cdots$
  is \emph{winning for Player~$2$} if
  there exists a \kl{partially perfect half space} $\tuple{\vec{g}_1, \ldots, \vec{g}_k}$ with $k > 0$ 
that is a prefix of $\vec{H}_i$ for all sufficiently large~$i$, s.t.\ $\limsup_{n} \sum_{j=1}^n\vec{w}_j \cdot \vec{g}_k = -\infty$
 and, for all $1 \leq \ell<k$,
$\liminf_{n} \sum_{j=1}^n \vec{w}_j \cdot \vec{g}_\ell < +\infty$.
\end{definition}  
\AP
Observe that whether Player~$2$
wins from $\tuple{v, \vec{H}}$ does not depend on~$\vec{H}$, hence we
say that Player~$2$ ""wins from $v$@win PHS"" if there exists $\vec H\in\?H$
such that he "wins@win PHS" from $(v,\vec H)$---equivalently, he "wins@win PHS" from
$(v,\vec H)$ for all $\vec H\in\?H$---, and similarly for Player~$1$.

\AP
Given a finite path
\begin{equation*}
 P\eqdef(v_0,\vec H_0)\xrightarrow{\vec w_1}(v_1,\vec
H_1)\cdots(v_{n-1},\vec H_{n-1})\xrightarrow{\vec w_n}(v_n,\vec H_n)
\end{equation*}
in a \kl{perfect half space game}, we denote by
$\intro*\phspath(P)\eqdef\phs_{0\leq i\leq n}\vec H_i$ the
longest \kl{partially perfect half space} that agrees with all
the \kl{perfect half spaces} seen along the path. We also inherit the
notation $\weight(P)\eqdef\sum_{i=1}^n\vec w_i$ that accounts for the
sum of the \kl{weights} in $P$.  \AP
We say that ""$P$ is winning for
Player~$1$@winning play PHS"" if $\weight(P)\cdot\phspath(P)\succeq\vec 0$.
Similarly, \reintro[winning play PHS]{$P$ is winning for Player~$2$} if
$\weight(P)\cdot\phspath(P)\prec\vec 0$.  Note that when $P$ is in fact a
cycle, then its infinite iteration is "winning for a player@winning play PHS" if and only
if the cycle is winning for them according to this definition.

\begin{example}
\label{ex:phsg}\AP
Player~$2$ wins the \kl{perfect half space game} on the graph
of \cref{ex:wgg} from any vertex %
by choosing the perfect half space $\vec H_L$ from \cref{ex:phs} when
going to $v_L$ and $\vec H_R$ when going to $v_R$.  Indeed, either
Player~$1$ eventually only uses the left (blue) cycle, in which case
$(\vec g_1,\vec g_2)\eqdef\vec H_L$ itself can be used as witness
in \cref{eq-phs}, or she eventually only uses the right (violet)
cycle, in which case $(\vec g_1,\vec g_2)\eqdef\vec H_R$ fits, or she
alternates infinitely often between $v_R$ and $v_L$ (using the cyan
cycle), in which case the partially perfect half space
$(\vec g_1)\eqdef((1,1))$ is a witness of his victory.
\end{example}

\section{Solving Perfect Half Space Games}\label{sec-lmp}
As an intermediate step towards the proof of our determinacy and
complexity results for \kl{perfect half space games} (\cref{th:tr.lmp}), we
employ another winning condition introduced in~\citep{colcombet17}:
that of \kl{lexicographic energy games}.  We start by presenting a
proof of their \kl{positional determinacy}, and an upper bound for their
decision problem using the state-of-the-art results of \citet{comin16}
for \kl{mean-payoff games}.  We then proceed to show how \kl{perfect half space
games} can be reduced to \kl{lexicographic energy ones} in
\cref{sub-tr-lmp}.

\subsection{Solving "Lexicographic Energy Games"}\label{sub-lmp}
{\makeatletter
\renewcommand{\@IEEEsectpunct}{\ }%
\makeatother
\subsubsection{\kl{Lexicographic Energy Games}}\citep{colcombet17}}%
\phantomintro{Lexicographic Energy Games}
are
played on \kl{multi-weighted game graphs} $(V, E, d)$, as described
in \cref{sec-phs}.  An infinite play $v_0\xrightarrow{\vec
  w_1}v_1\xrightarrow{\weight_2}\cdots$ is \emph{winning for
  Player~$2$} if there exists $1\le k\le d$ s.t.\ $\limsup_{n}
\sum_{j=1}^n\vec{w}_j(k) = -\infty$ and, for all $1 \leq \ell<k$,
$\liminf_{n} \sum_{j=1}^n \vec{w}_j(\ell) <
+\infty$.\footnote{\kl{Lexicographic energy games} bear a superficial
resemblance to two different definitions of lexicographic mean-payoff
games, due respectively to \citet{bloem09} and to \citet{bruyere14}.
However, the definition that would best match \kl{lexicographic energy
games} would be multi-dimensional `pointwise' lexicographic mean-payoff
games, which do \emph{not} enjoy positional determinacy, and all these
definitions are unfit for our purposes.}

Put differently, "lexicographic energy games" are akin to "perfect half
space games", except that the same full perfect half space $(-\vec
e_1,\dots,-\vec e_d)$ is associated to every vertex of the game graph,
where $\vec e_i$ for $1\leq i\leq d$ denotes the unit vector with $1$
in coordinate $i$ and $0$ everywhere else.

\begin{example}\label{ex-lmp}
  Let us consider the \kl{multi-weighted game graph} of \cref{ex:wgg}.
  Player~$1$ wins the lexicographic energy game from any initial
  vertex, by moving to $v_L$ and looping on the left (blue) loop.
\end{example}

\subsubsection{Strategies}
 \AP A strategy for a player is \intro[positional strategy]{positional} if,
from each of her vertices, the player using it
always chooses the same outgoing edge, no
matter where the play started or how it evolved so far.  We say that a
game is \intro{positionally determined} if the two players have
"positional strategies"~$\sigma$ and~$\tau$, respectively, such that for
every vertex~$v \in V$, either $\sigma$ is winning for Player~$1$ from~$v$,
or $\tau$ is winning for Player~$2$ from~$v$.

\subsubsection{Reduction to Mean-Payoff Games}
A \intro{mean-payoff game} is played on a \intro{weighted game graph},
i.e. a \kl{$1$-dimensional weighted game graph} $(V,E,1)$, and is
denoted~$\intro*\MPgame(V,E)$.  From an infinite play
$v_0 \xrightarrow{u_1} v_1 \xrightarrow{u_2} v_2 \cdots$, Player~$1$
(`Max') gains a payoff $\liminf_{n \to \infty} (u_1 + \cdots + u_n) /
n$, whereas Player~$2$ (`Min') loses a payoff $\limsup_{n \to \infty}
(u_1 + \cdots + u_n) / n$.  A strategy for Max is \emph{optimal} for
her if by following it she is guaranteed to gain at least as much as
when using any other strategy, and optimal strategies for Min are
defined symmetrically.  By the positional determinacy of mean-payoff
games~\citep{zwick96}, there exist positional optimal strategies for
both players, yielding the same payoff for both from each initial
vertex, called the \emph{value} of the vertex.

A strategy for Max is \emph{winning} from some initial vertex if by
following it she is guaranteed to gain at least $\geq 0$, and a
strategy for Max is winning if by following it he is guaranteed to
lose at least $<0$.  Note that not every winning strategy for Min
needs to be optimal, but that if she wins then any optimal strategy is
winning: Min wins the game if and only if the value of the initial
vertex is $\geq 0$, and Max wins otherwise.

For a \kl{multi-weighted game graph} $(V, E, d)$, and for every~$i$,
$1 \leq i \leq d$, let the set $E(i)$ consist of the edges 
$v \xrightarrow{\weight(i)} v'$ where 
$v \xrightarrow{\weight} v' \in E$. 

\begin{theorem}
\label{th:lmp}
\begin{enumerate}[(i)]
\item \kl{Lexicographic energy games} are \kl{positionally determined}.
\item\label{th:lmp:complexity} There is an algorithm for solving \kl{lexicographic energy games}
  whose running time is in
  $O\left(|V|^{d+1}\cdot|E|\cdot\prod_{i=1}^d\|E(i)\|\right)$.
\end{enumerate}
\end{theorem}

We start by describing a translation from \kl{lexicographic energy
games} to \kl{mean-payoff games}, similar to the classical translation
from parity games~\citep{puri1995}: the idea is to write the
$d$-dimensional weights into a single weight by shifting the most
significant components by appropriate amounts.  We define accordingly
the sets of weighted edges $E^{(i)}$ for $i = d, d-1, \dots, 1$ as
follows:
\begin{itemize}
\item $E^{(d)} \eqdef E(d)$;
\item for all $i = d-1, d-2, \dots, 1$, and for all
$v \xrightarrow{\weight} v' \in E$, if $v \xrightarrow{r_{i+1}} v' \in
E^{(i+1)}$ then $v \xrightarrow{r_i} v' \in E^{(i)}$, where
$r_i \eqdef \weight(i) \cdot \left(|V| \cdot \|E^{(i+1)}\| + 1\right)
+ r_{i+1}$.
\end{itemize}
We will argue directly that \kl{positional optimal strategies} for the
two players in the \kl{mean-payoff game} $\MPgame(V, E^{(1)})$
witness \kl{positional determinacy} of the \kl{lexicographic energy
game} $\LexEngame(V, E, d)$.
\begin{figure}[tbp]
  \centering
  \begin{tikzpicture}[auto,on grid]
    \node[triangle](L){$v_L$};
    \node[triangle, right=3.6cm of L](R){$v_R$};
    \node[square,above right=.5 and 1.8 of L](A){};
    \node[square,below right=.5 and 1.8 of L](B){};
    \node[square, left=1.5 of L](L1){};
    \node[square, right=1.5 of R](R1){};
    \path[->,every node/.style={font=\footnotesize,inner sep=1pt}]
      (L) edge[bend left=15,draw=black!40!cyan,pos=0.8] node{$0$} (A)
      (A) edge[bend left=15,draw=black!40!cyan,pos=0.2] node{$-7$} (R)
      (R) edge[bend left=15,draw=black!40!cyan,pos=0.6] node{$0$} (B)
      (B) edge[bend left=15,draw=black!40!cyan,pos=0.35] node{$-1$} (L)
      (L) edge[bend left,draw=black!40!blue] node{$6$} (L1)
      (L1) edge[bend left,draw=black!40!blue] node{$0$} (L)
      (R) edge[bend left,draw=black!40!violet] node{$-6$}(R1)
      (R1) edge[bend left,draw=black!40!violet] node{$0$}(R)
  ;
  \end{tikzpicture}
  \caption{\label{fig:mpg}The weighted game graph $(V,E^{(1)})$
  constructed from the graph of \cref{fig-mean}.}
\end{figure}
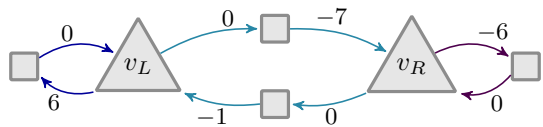
\begin{example}
  The weighted game graph obtained from the multi-weighted game graph
  of \cref{ex:wgg} is depicted in \cref{fig:mpg} (indeed
  $|V|\cdot\|E^{(2)}\|+1=7$).  Max has a positional optimal strategy
  consisting in moving to $v_L$ and using the left (blue) loop; every
  vertex has value~$6$.
\end{example}

The outcome of this encoding of $d$-dimensional weights in $E^{(1)}$
is the following, easy to establish, proposition.

\begin{proposition}
\label{prop:sign-pres}
  The total weight of a simple cycle in the \kl{multi-weighted game
  graph} $(V, E, d)$ is $\prec \vec 0$ 
  (or $= \vec 0$, or $\succ \vec 0$, respectively)
  if and only if the total \kl{weight} of the cycle in the \kl{weighted game
  graph} $(V, E^{(1)})$ is negative (or zero, or positive,
  respectively).
\end{proposition}
  
In order to show the "positional determinacy" of "lexicographic
energy games", we rely on the following lemma proven in
the appendix.

\begin{restatable}{lemma}{lmpgpositional}
\label{lem:lmpg-positional}
  If the value of the \kl{mean-payoff game} $\MPgame(V, E^{(1)})$ is non-negative
  (negative, resp.) at a vertex~$v$, then by using a \kl{positional optimal strategy} from that \kl{mean-payoff game}, Player~$1$ (Player~$2$, resp.) wins the corresponding \kl{lexicographic energy game} $\LexEngame(V, E, d)$ from~$v$.
\end{restatable}

\begin{proof}[Proof of \cref{th:lmp}]
By~\cref{lem:lmpg-positional}, in order to compute a positional
winning strategy for one the players in a \kl{lexicographic energy
game} $\LexEngame(V, E, d)$, it suffices to find a positional optimal
strategy in the corresponding \kl{mean-payoff game} $\MPgame(V,
E^{(1)})$.  This entails the \kl{positional determinacy}
of \kl{lexicographic energy games} (cf., e.g.,~\citep{zwick96}).
Regarding complexity, the state-of-the-art algorithm for
solving \kl{mean-payoff games} due to \citet{comin16} runs in time
$O\left(|V|^2 \cdot |E| \cdot \|E\|\right)$.  Observe that $|E^{(1)}|
= |E|$ and $\|E^{(1)}\| =
O\left(|V|^{d-1} \cdot \prod_{i=1}^d \|E(i)\|\right)$, and hence the
algorithm of \citeauthor{comin16} can be used to solve lexicographic
energy games in time $O\!\!\left(\!|V|^{d+1} \cdot
|E| \cdot \prod_{i=1}^d\!\|E(i)\|\!\right)\!$.
\end{proof}

\begin{figure*}[tbp]
  \centering
  \begin{tikzpicture}[auto,on grid]
    \shade[left color=violet!30,right color=white] (2.6,1.1) --
  (1,-1.1) -- (7,-1.1) -- (7,1.1) -- cycle;
    \shade[right color=blue!22,left color=white] (-4.4,1.1) --
  (-4.4,-1.1) -- (1,-1.1) -- (2.6,1.1) -- cycle;
    \draw[dashed,black!40,thick] (2.6,1.1) -- (1,-1.1);
    \node[color=black!40!blue] at (-2.4,0) {$\vec H_L$};
    \node[color=black!40!violet] at (6,0) {$\vec H_R$};
    \node[triangle](L){$v_L$};
    \node[triangle, right=3.6cm of L](R){$v_R$};
    \node[square,above right=.5 and 1.8 of L](A){};
    \node[square,below right=.5 and 1.8 of L](B){};
    \node[square, left=1.5 of L](L1){};
    \node[square, right=1.5 of R](R1){};
    \path[->,every node/.style={font=\tiny,inner sep=.5pt}]
      (L) edge[bend left=15,draw=black!40!cyan,pos=0.8] node{$\tuple{0,0,0,0}$} (A)
      (A) edge[bend left=15,draw=black!40!cyan,pos=0.2] node{$\tuple{0,-1,1,1}$~} (R)
      (R) edge[bend left=15,draw=black!40!cyan,pos=0.6] node{$\tuple{0,0,0,0}$} (B)
      (B) edge[bend left=15,draw=black!40!cyan,pos=0.35] node{$\tuple{0,-1,1,1}$} (L)
      (L) edge[bend left,draw=black!40!blue] node{$\tuple{0,0,0,-2}$} (L1)
      (L1) edge[bend left,draw=black!40!blue] node{$\tuple{0,0,0,0}$} (L)
      (R) edge[bend left,draw=black!40!violet] node{$\tuple{0,0,0,-2}$}(R1)
      (R1) edge[bend left,draw=black!40!violet] node{$\tuple{0,0,0,0}$}(R)
  ;    
  \end{tikzpicture}
  \caption{The translation of the graph from \cref{fig-mean}
  to lexicographic energy games.\label{fig-lmp}}%
\end{figure*}
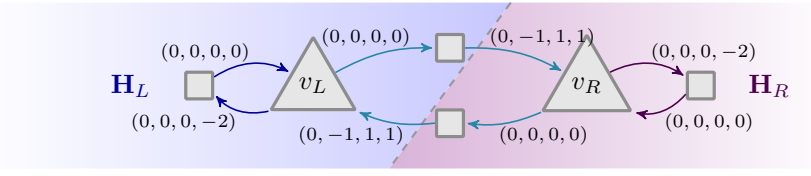

\subsection{Translation to Lexicographic
  Energy Games}\label{sub-tr-lmp}

We now reduce \kl{perfect half space games} to \kl{lexicographic energy games}.
Given a \kl{perfect half space game} played on a \kl{$d$-dimensional
multi-weighted game graph}, the idea is to play a \kl{lexicographic energy game}
on a \kl{$2d$-dimensional game graph}, where the extra dimensions are used to
penalise Player~$2$ for changing of \kl{perfect half space}.
\newcommand\HS{\vec{H}} \newcommand\PHS{\vec{H}}

\subsubsection{\kl{Flag Vectors} and \kl{Interleavings}}
\AP
For any \kl{$d$-dimensional perfect half spaces} $\vec{H}$ and
$\vec{H}'$, let the \intro{flag vector} $\vec{e}_{\vec{H}, \vec{H}'}$
be defined for all $i=1,\dots,d$ by
$\vec{e}_{\vec{H}, \vec{H}'}(i)\eqdef 0$ if the $i$-th coordinates of
$\vec{H}$ and $\vec{H}'$ are identical, and
$\vec{e}_{\vec{H}, \vec{H}'}(i)\eqdef 1$ otherwise.
For any $d$-dimensional vectors $\vec{a}$ and $\vec{b}$,
let $\vec{a} \shuffle \vec{b}$ be their \intro{interleaving}
$\tuple{\vec{a}(1), \vec{b}(1), \ldots, \vec{a}(d), \vec{b}(d)}$.

\subsubsection{Translation}
\AP
We write $\tuple{\widehat{V}, \widetilde{E}, 2d}$ 
for the \kl{weighted game graph} obtained from 
$\tuple{\widehat{V}, \widehat{E}, d}$ by doubling the dimension, 
where the even indices of \kl{weights} in $\widetilde{E}$ 
contain the corresponding \kl{weights} from $\widehat{E}$
but normalised with respect to the current \kl{perfect half space},
and the odd indices are occupied by \kl{flag vectors} that 
penalise Player~$2$ for changing the \kl{perfect half spaces}.
More precisely, $\widetilde{E}$ is the set of all
$\tuple{v, \vec{H}}
 \xrightarrow{\vec{e}_{\vec{H}, \vec{H}'} \shuffle (\vec{w} \cdot \vec{H})}
 \tuple{v', \vec{H}'}$
such that $\tuple{v, \vec{H}} \xrightarrow{\vec{w}} \tuple{v', \vec{H}'}$
is in~$\widehat{E}$.

\AP
Let $\intro*\LexEngame\tuple{\widehat{V}, \widetilde{E}, 2d}$ denote the
lexicographic energy game played on the multi-weighted game graph
$(\widehat{V}, \widetilde{E}, 2d)$.

\begin{example}
  We depict in \cref{fig-lmp} a fragment of the translated game graph
  $(\widehat{V},\widetilde{E},2d)$ for the \kl{perfect half space game}
  from \cref{ex:phsg}.  The vertices on the left of
  the median dashed line are all paired with $\vec H_L$, while those
  on the right are paired with $\vec H_R$.  The \kl{flag vector} $\vec
  e_{\vec H_L,\vec H_R}=(0,1)=\vec e_{\vec H_R,\vec H_L}$ is
  "interleaved" with the normalised vectors on the two middle edges
  entering $v_R$ and~$v_L$.  

  In contrast to \cref{ex-lmp}, Player~$1$ now loses the lexicographic
  energy game in \cref{fig-lmp}.  Indeed, if she plays the middle
  simple cycle (in cyan) infinitely often, then the energy on the
  first coordinate converges to~$0$ and the energy in the second
  coordinate diverges to~$-\infty$.  Otherwise (i.e., if the number of
  occurrences of the middle cycle is bounded), the energy in the first
  three coordinates does not diverge and the energy in the fourth
  coordinate diverges to~$-\infty$.
\end{example}

The correctness of this translation is a direct consequence of the
definitions, as shown in the following lemma proven in the appendix.
\begin{restatable}{lemma}{lemtrlmp}
\label{th:tr.lmp.1}
  The winning strategies of Player~$i$, $i\in\{1,2\}$,
  are the same in  
  $\PHSgame\tuple{\widehat{V}, \widehat{E}, d}$ and
  $\LexEngame\tuple{\widehat{V}, \widetilde{E}, 2d}$.
\end{restatable}

\AP
Define a strategy $\tau$ for Player~$2$ in the \kl{perfect half space
game} $\PHSgame\tuple{\widehat{V}, \widehat{E}, d}$ to be
\intro{(perfect half space) oblivious at~$v$} for $v\in V_2$ if it chooses the same
move in $(v,\HS)$ for all~$\HS$.  It is \intro{perfect half space
oblivious} if it is "oblivious at all vertices"~$v\in V_2$.  We are now
ready to prove the main theorem of this section.

\begin{theorem}
\label{th:tr.lmp}
\begin{enumerate}[(i)]
\item
\label{th:tr.lmp.2}
  There is an algorithm for solving \kl{perfect half space games}
  whose running time is in
  $O\left((3 |V| \cdot \|E\|)^{2(d + 1)^3}\right)$.
\item
\label{th:tr.lmp.3}
  If Player~$2$ has a winning strategy in the \kl{perfect half
  space game} $\PHSgame\tuple{\widehat{V}, \widehat{E}, d}$, then he has 
  one that is \kl{perfect half space oblivious}.
\end{enumerate}
\end{theorem}
\begin{proof}[Proof of \cref{th:tr.lmp}(\ref{th:tr.lmp.2})]
The upper bound on the running time is a consequence of
\cref{th:tr.lmp.1} and \cref{th:lmp}.\ref{th:lmp:complexity}.  Observe that the vertex set is of size $|\widehat
V|=|V|\cdot|\?H|\leq |V|\cdot (2 |V| \cdot \|E\| + 1)^{d^2}\leq
(3|V|\cdot\|E\|)^{d^2+1}$.  
Regarding the norms, $\|\widetilde
E\|=\max \{\|\weight\cdot\vec H\|\,:\,v\xrightarrow{\weight}v'\in E,\vec
H\in\?H\}$, hence $\|\widetilde E\|\leq d\cdot
|V|\cdot\|E\|^2\leq (3|V|\cdot\|E\|)^{2+\log d}$.  Hence a time bound
in $O\left((3|V|\cdot\|E\|)^m\right)$ where
$m=(d^2+1)(2d+3)+2d(2+\log d)
\leq 2(d+1)^3$.
\end{proof}
\begin{proof}[Proof of \cref{th:tr.lmp}(\ref{th:tr.lmp.3})]
\newcommand\thePHSGame{\PHSgame\tuple{\widehat{V}, \widehat{E}, d}}
The idea of the following proof is to show that, for any vertex of the
weighted game graph winning for Player~$2$, there is a `good' perfect
half space $\vec H$ such that following a positional strategy
$\tau_\vec H$ winning from $(v,\vec H)$ will also win from any
$(v,\vec H')$.

More formally, we prove by induction on~$k\leq|V_2|$ that there exists
a \kl{winning positional strategy}~$\tau$ for Player~$2$ which
is \kl{perfect half space oblivious at} $k$ distinct vertices
in~$V_2$.

The induction hypothesis obviously holds for~$k=0$ by using a
positional strategy in~$\thePHSGame$, which exists by
\cref{th:lmp,th:tr.lmp.1}.  For the induction step, let us suppose
that $\tau$ is a \kl{winning positional strategy for Player~$2$}
\kl{oblivious} at~$k<|V_2|$ distinct vertices $v_1,\dots,v_k\in V_2$.
Let $v$ be another $v\in V_2$ distinct from $v_1,\dots,v_k$ vertices;
$\tau$ and $v$ are now fixed for the remainder of the proof.

For all \kl{perfect half spaces}~$\HS$, let us denote by~$\tau_\HS$
the strategy~$\tau$ modified in such a way that it behaves in
$(v,\HS')$ as in $(v,\HS)$ for all~$\HS'\neq\HS$.  The result is still
a valid strategy (by definition of the \kl{perfect half space game})
and is of course \kl{oblivious at $v$} as well as at $v_1,\dots,v_k$.
We want to show that there exists $\vec H$ such that $\tau_\vec H$
fulfils the induction hypothesis.  This is the case for any $\vec H$
if~$v$ is not in the winning region.  We shall therefore assume that
$v$ is in the winning region for Player~$2$; thus $\tau$ is winning
from every $(v,\vec H)$ but might use different moves depending on
$\vec H$.

\paragraph{Good Perfect Half-Spaces}
Let us call a \kl{perfect half space}~$\HS$ \intro[good]{good (for
$\tau$ and $v$)} if $\tau_\HS$ is winning for Player~$2$ starting in
$(v,\HS)$, and "bad@good" otherwise.  As "shown in the appendix@proof:cl-gphs", there
must exist a \kl{good} \kl{perfect half space}, as otherwise
Player~$2$ would not win from $v$.

\begin{restatable}{claim}{clgphs}\label{cl-gphs}
 There exists a \kl{good} \kl{perfect half space}.
\end{restatable}

\paragraph{A Winning Strategy $\tau_\vec H$}
Let~$\HS$ be a \kl{good perfect half space} that exists according to
\cref{cl-gphs}.  Let us show that~$\tau_\HS$ fulfils the condition of
the induction hypothesis.  As already mentioned, it is \kl{oblivious
  at} $\{v,v_1,\dots,v_k\}$.  We have to prove that it is winning.
For this, let us consider any play consistent with~$\tau_\HS$ starting
from some $(v',\HS')$ in the winning region for Player~$2$.  Two cases
can happen.  Either this play does not visit the vertex~$v$, and in
this case it was already a run consistent with~$\tau$, and hence it is
winning for Player~$2$.  Otherwise it visits~$v$, and after that point
it continues in a way consistent with $\tau_\HS$ starting
from~$(v,\HS)$, and hence is winning for Player~$2$ since~$\HS$ is
\kl{good}.  This establishes the induction hypothesis, and thus
completes the proof of \cref{th:tr.lmp}(\ref{th:tr.lmp.3}).
\end{proof}

\section{Bounding Games}\label{sec-bnd}

\newcommand\nokl[1]{#1}

In this section, we define \kl{bounding games} (as introduced
in \citep{jurdzinski15}) and show how these can be reduced
to \nokl{perfect half space games} (\cref{th:eq.bounding}
below). \Cref{cor-bnd} then summarises our knowledge
about "bounding games".

\AP For a \kl{weighted game graph} $\tuple{V, E, d}$, we denote by
$\intro*\Bndgame\tuple{V, E, d}$ the \intro{bounding
  game} in which Player~$1$ (`Guard') strives
to contain the total weight within some $d$-dimensional hypercube,
while Player~$2$ (`Fugitive') attempts to escape.  More precisely, an
infinite play $v_0 \xrightarrow{\vec{w}_1} v_1 \xrightarrow{\vec{w}_2}
v_2 \cdots$ is winning for Player~$1$ if and only if the set
$\{\|\sum_{i = 1}^n \vec{w}_i\| \,:\, n \in \+N\}$ of norms of total
weights of all finite prefixes of the play is bounded.

\begin{example}
  Consider again the \kl{multi-weighted game graph} of \cref{ex:wgg}.
  Observe that Player~$1$ cannot choose to play solely in the left
  (blue) cycle, as the accumulated weights would drift towards
  $(+\infty,-\infty)$; a similar argument holds with the right
  (violet) cycle.  Hence, she must somehow balance the effect of the
  two cycles by switching infinitely often between $v_L$ and $v_R$,
  but the effect of the middle (cyan) cycle then makes the simulated
  weights drift towards $(-\infty,-\infty)$.  In fact, by the upcoming
  \cref{th:eq.bounding} and as seen in \cref{ex:phsg}, Player~$2$ wins
  this game.
\end{example}

\begin{theorem}
\label{th:eq.bounding}
Let $\tuple{V,E,d}$ be a \kl{multi-weighted game graph},
$v$ be a vertex in $V$,
and $i\in\{1,2\}$.  Player~$i$ wins the \kl{bounding game}
$\Bndgame\tuple{V, E, d}$ from $v$ if and only if Player~$i$ wins the
\kl{perfect half space game} $\PHSgame\tuple{\widehat{V}, \widehat{E}, d}$
from $v$.
\end{theorem}
By \cref{th:tr.lmp}, \kl{perfect half space games} are determined, hence we
can focus on Player~$2$.  One implication is straightforward:
a winning strategy for Player~$2$ in $\PHSgame\tuple{\widehat{V},
  \widehat{E}, d}$ also wins $\Bndgame\tuple{V, E, d}$ when ignoring
the \kl{perfect half spaces}.  Note that this translates an \kl{oblivious
strategy} in $\PHSgame\tuple{\widehat{V}, \widehat{E}, d}$ into a 
\kl[positional strategy]{positional one} in $\Bndgame\tuple{V, E, d}$.
\begin{lemma}\label{cl-phs-bnd}
  If Player~$2$ wins $\PHSgame\tuple{\widehat{V}, \widehat{E}, d}$ from
  $v$, then he wins $\Bndgame\tuple{V, E, d}$ from~$v$ with the same
  strategy (where perfect half spaces are projected away).
\end{lemma}
\begin{proof}[Proof sketch]
Let Player~$2$ follow a winning strategy for the \kl{perfect half space game},
projected onto the arena of the bounding game,
and consider any resulting play.
By the winning condition of the former game,
the total weights have unbounded distances from some hyperplane,
and so have unbounded norms.
\end{proof}

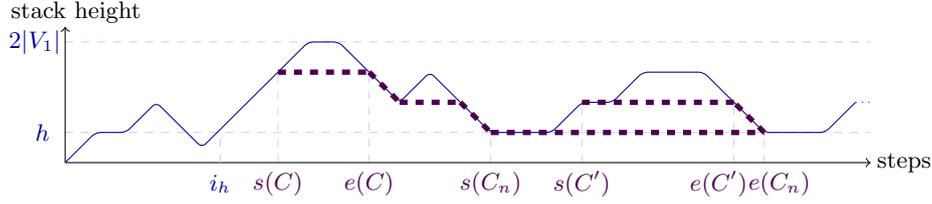
\begin{figure*}[tbp]
  \centering
  \begin{tikzpicture}[scale=.4,every node/.style={font=\footnotesize}]
    \draw[color=blue!20,dashed] (0,4) -- (26.5,4);
    \draw[color=blue!20,dashed](0,1) -- (26.5,1);
    \draw[color=violet!20,dashed] (7,0) -- (7,3) (17,0) -- (17,2)
    (23,0) -- (23,1) (14,0) -- (14,1) (22,0) -- (22,2) (10,0) -- (10,3);
    \draw[->] (0,0) -- (26.5,0);
    \draw[->] (0,0) -- (0,4.5);
    \node at (2.2ex,5) {stack height};
    \node at (27.6,0) {steps};
    \draw[color=black!40!blue,rounded corners=2pt]
      (0,0) -- (1,1) -- (2,1) -- (3,2) -- (4.5,0.5) -- (5,1)
            -- (8,4) -- (9,4) -- (11,2) -- (12,3) -- (14,1)
            -- (16,1) -- (17,2) -- (18,2) -- (19,3) -- (20,3)
            -- (21,3) -- (23,1) -- (25,1) -- (26,2);
    \draw[color=black!40!blue,dotted] (26,2) -- (26.5,2);
    \node[color=black!40!blue] at (-1,4) {$2|V_1|$};
    \node[color=black!40!blue] at (-.75,1) {$h$};
    \draw[color=blue!20,dashed] (5.1,0) -- (5.1,1);
    \node[color=black!40!blue] at (5.1,-.7) {$i_h$};
    \draw[color=black!40!violet,line width=2pt,rounded corners=2pt,dashed] 
               (10,3) -- (11,2)    (13,2) -- (14,1)    (16,1)    (23,1) -- (22,2)    (18,2);
    \draw[color=black!40!violet,line width=2pt,rounded corners=2pt,dashed]
      (7,3) -- (10,3)    (11,2) -- (13,2)    (14,1) -- (16,1) -- (23,1)    (22,2) -- (18,2) -- (17,2);
    \node[color=black!40!violet] at (7,-.7) {$s(C)$};
    \node[color=black!40!violet] at (10,-.7) {$e(C)$};
    \node[color=black!40!violet] at (17,-.7) {$s(C')$};
    \node[color=black!40!violet] at (21.5,-.7) {$e(C')$};
    \node[color=black!40!violet] at (14,-.7) {$s(C_n)$};
    \node[color=black!40!violet] at (23.5,-.7) {$e(C_n)$};
  \end{tikzpicture}
  \caption{Stack heights in the proof of \cref{cl-bnd}.\label{fig-rhon}}
\end{figure*}
It remains therefore to establish the converse implication in order to
complete the proof of \cref{th:eq.bounding}.
\begin{restatable}{lemma}{bndphs}\label{cl-bnd-phs}
  If Player~$2$ wins $\Bndgame\tuple{V, E, d}$ from~$v$, then he wins 
$\PHSgame\tuple{\widehat{V}, \widehat{E}, d}$ from~$v$.
\end{restatable}
\AP The proof of this lemma 
relies on \citep[Lemma~5.5]{jurdzinski15}---the most involved
result in that paper---, which shows how to construct a winning
strategy for Player~$1$ from $v$ in $\Bndgame\tuple{V, E, d}$ from a
\kl{winning strategy} in a \emph{\kl[first-cycle game]{first-cycle}} variant
$\intro*\FCgame\tuple{\widehat{V}, \widehat{E}, d}$ of
$\PHSgame\tuple{\widehat{V}, \widehat{E}, d}$ from $v$.  As these
first-cycle games are determined, this entails that, if Player~$2$
wins from $v$ in $\Bndgame\tuple{V, E, d}$, then he also wins from $v$
in the first-cycle game $\FCgame\tuple{\widehat{V}, \widehat{E}, d}$, and it
remains to show how to build a winning strategy for him in
$\PHSgame\tuple{\widehat{V}, \widehat{E}, d}$.  The reasoning itself is
surprisingly subtle, and similar to the one employed in the proof of
\citep[Lemma~5.3]{jurdzinski15}.
\begin{proof}[Proof of \cref{cl-bnd-phs}]\AP
  By \citep[Lemma~5.5]{jurdzinski15}, there exists a winning strategy
  $\sigma$ for Player~$2$ from some $\tuple{v, \vec{H}}$ in the
  following \intro{first-cycle game} $\FCgame\tuple{\widehat{V}, \widehat{E}, d}$:
  \begin{enumerate}
  \item
    the game finishes as soon as the play has a suffix
    $C=\tuple{v_0, \vec{H}_0} \xrightarrow{\vec{w}_1}
    \tuple{v_1, \vec{H}_1} \xrightarrow{\vec{w}_2} \cdots\xrightarrow{\vec{w}_n}
    \tuple{v_n, \vec{H}_n}$
    such that $v_0 = v_n \in V_1$;
  \item\label{fcg} Player~$2$ wins if and only if $\vec{H}_0 =
    \vec{H}_n$ and $C$ is winning for him: the total weight $\weight(C)\eqdef\vec{w}_1 + \cdots + \vec{w}_n$ of the cycle is in the \kl{partially perfect half space}
    defined by the longest common prefix,  i.e.\
  $\weight(C)\cdot\phs(C)\prec\vec 0$
  (recall that $\phspath(C)\eqdef\phs_{1\leq i\leq n}\vec H_i$).
\end{enumerate}

Let $\sigma^*$ denote the strategy for Player~$2$ from $\tuple{v,
  \vec{H}}$ in $\PHSgame\tuple{\widehat{V}, \widehat{E}, d}$ that amounts
to following $\sigma$ and repeatedly cutting out the winning cycles.
We want to show that $\sigma^*$ is winning: consider for this a play
$\tuple{v_0, \vec{H}_0} \xrightarrow{\vec{w}_1} \tuple{v_1, \vec{H}_1}
\xrightarrow{\vec{w}_2} \cdots$ consistent with $\sigma^*$ starting
from $v_0=v$.

\AP Let us consider the \intro{$V_1$ cycle decomposition} of this play: the
latter is the infinite sequence of `$V_1$-simple' cycles $C$
obtained by pushing the triples of visited vertices and \kl{perfect half
spaces} and indices $(v_0,\vec H_0,0),(v_1,\vec H_1,1),\dots$ onto a
stack, and as soon as we push a pair $(v_e,\vec H_e,e)$ with an
element $(v_s,\vec H_s,s)$ with $v_s=v_e\in V_1$ already present in
the stack, we pop the cycle $C$ thus formed from the stack and
push $(v_e,\vec H_e,e)$ back on top.  We call the indices
$s(C)\eqdef s$ and $e(C)\eqdef e$ the \emph{start} and
\emph{end} of the cycle, and denote by $\phs(C)$ and $\vec
w(C)$ the longest common prefix of its perfect half spaces and total
weight respectively.  Because $\sigma$ is winning in
$\FCgame\tuple{\widehat{V}, \widehat{E}, d}$, all the cycles $C$ formed in
the cycle decomposition satisfy condition~\ref{fcg} above, hence $\vec
H_{s(C)}=\vec H_{e(C)}$ and $\weight(C)\cdot\phs(C)\prec\vec 0$.

\AP Let us now consider the longest $\vec P$ such that there exists a
sufficiently large index $i_\vec P$ such that $\vec
P=\phs_{s(C)\geq i_\vec P}(\phspath(C))$.  We call a
"partially perfect half space representation" $\vec H$ \intro{recurring}
if $\vec H=\phs(C)$ for infinitely many cycles $C$ in the $V_1$
cycle decomposition of our play; such a vector $\vec H$ is necessarily
non-empty.
\begin{claim}\label{cl-bnd}
  $\vec P$ is \kl{recurring}.
\end{claim}
\begin{proof}[Proof of \cref{cl-bnd}]
  We reason on the height of the stack used to construct the \kl{$V_1$
  cycle decomposition} of the play.  Let us call $\rho_i$ the stack at
  step $i$.  Since its height $|\rho_i|$ is bounded by $2|V_1|$, there
  is a smallest height $h$ that occurs infinitely often, and a minimal
  index $i_h$ such that $h$ is the minimal height in the infinite
  suffix starting from $i_h$.  We depict the stack heights along the
  play in blue in \cref{fig-rhon}.
  
  \AP Let us call a \intro{downward path} a sequence of cycles
  $C_1,\dots,C_n$ such that, for all $1\leq i<n$, $C_i$ and $C_{i+1}$
  are either two successive cycles with
  $|\rho_{s(C_i)}|>|\rho_{s(C_{i+1})}|$ or two cycles (not necessarily
  successive) with $e_{C_i}=s_{C_{i+1}}$.  Observe that in both cases,
  they visit a common perfect half space $\vec H_{e_{C_i}}$, hence
  $\phs(C_i)$ and $\phs(C_{i+1})$ are comparable for the prefix
  ordering.
  
Assume there are two \kl{recurring} "representations@representation PHS" of partially
perfect half spaces $\vec H$ and $\vec H'$.  Let us show that they
have a common prefix that is also \kl{recurring}.  For this, consider
two occurrences $\phs(C)=\vec H$ and $\phs(C')=\vec H'$ of $\vec H$
and $\vec H'$ with $i_h<s(C)<s(C')$.  As shown by the thick dashed
violet line in \cref{fig-rhon}, and since a stack height of $h$ occurs
infinitely often, there must be two \kl{downward paths}
$C=C_1,\dots,C_n$ resp.\ $C'=C_{n+m},\dots,C_{n}$ from $C$ resp.\ $C'$
to a single cycle $C_n$.  Thus the sequence
$C=C_1,\dots,C_n,\dots,C_{n+m}=C'$ is such that, for all $1\leq
i<n+m$, $\phs(C_i)$ and $\phs(C_{i+1})$ are comparable for the prefix
ordering.  The set $\{\phspath(C_i)\,:\,1\leq i\leq n+m\}$ is a finite
meet-semilattice for the prefix ordering, thus with a bottom element
$\vec G\leqpref\vec H,\vec H'$.  As there are infinitely many such
pairs of occurrences of the \kl{recurring} $\vec H$ and $\vec H'$ and
finitely many different such $\vec G$ with $\|\vec G\|\leq
|V|\cdot\|E\|$, one of the latter must be \kl{recurring}.

To conclude the proof, assume now that $\vec P$ is not \kl{recurring} and
let us derive a contradiction.  Note that, for all cycles $C$ with
$s(C)\geq i_\vec P$, $\vec P\leqpref\phs(C)$.  Since $\vec
P$ is not \kl{recurring}, there must be two incomparable \kl{recurring} $\vec H$
and $\vec H'$, such that $\vec P\lpref\vec H$ and $\vec
P\lpref\vec H'$; we shall further assume that $\vec H$ is
minimal in length with this property.  By the previous argument, they
have a common prefix $\vec G\lpref\vec H,\vec H'$, which
is also \kl{recurring}, and which we shall also assume minimal in length.
Since $\vec H$ was chosen minimal, there is no \kl{recurring} $\vec G'$
incomparable with $\vec G$, and since $\vec G$ is minimal,
there is no \kl{recurring} $\vec G'\lpref\vec G$ either, hence
there exists an index $i$ such that $\vec
G=\phs_{s(C)\geq i}(\phspath(C))$.  As $\vec
P\leqpref\vec G$ and $\vec P$ was chosen of maximal
length with this property, $\vec P=\vec G$ is \kl{recurring}.
\end{proof}

Let us conclude the proof of \cref{cl-bnd-phs}.  Write $\vec P$ as
$(\vec p_1,\dots,\vec p_{|\vec P|})$.  For all cycles $C$ with
$s(C)\geq i_\vec P$, $\vec P\leqpref\phs(C)$ shows that $\vec
w(C)\cdot\vec P\preceq\vec 0$.  There are then $|\vec P|+1$ cases for
such cycles $C$: either there is $1\leq k\leq |\vec P|$ with $\vec
w(C)\cdot\vec p_k<0$ and $\weight(C)\cdot\vec p_\ell=0$ for all
$1\leq\ell\leq k$, or $\weight(C)\cdot\vec P=\vec 0$.  By
\cref{cl-bnd}, the $|\vec P|$ first cases occur (cumulatively)
infinitely often; let $k^*$ with $1\leq k^*\leq |\vec P|$ be the
smallest that does.  Then, as there are only finitely many occurrences
of cases $k<k^*$, and finitely many $\weight_j$ and $\vec H_j$ not
taken into account in the set of cycles $C$ with $s(C)\geq i_\vec P$,
$(\vec p_1,\dots,\vec p_{k^*})$ is a witness for \cref{eq-phs}:
Player~$2$ wins the play.
\end{proof}

By \cref{th:tr.lmp}, \cref{th:eq.bounding} and the proof of \cref{cl-phs-bnd},
we now have the following improvement over \citep[Corollary~3.2]{jurdzinski15}.

\begin{corollary}\label{cor-bnd}
\begin{enumerate}[(i)]
\item There is an algorithm for solving \kl{bounding games} whose running
  time is in $(|V| \cdot \|E\|)^{O(d^3)}$.
\item Player~$2$ has \kl{positional winning strategies} for \kl{bounding games}.
\end{enumerate}
\end{corollary}

\section{Multi-Dimensional \mbox{Energy Parity Games}}\label{sec-mep}

In this section, we define  \kl{multi-dimensional energy parity games}
(as introduced in \citep{chatterjee14})
as well as \kl{extended multi-dimensional energy games} (from~\citep{brazdil10}), and show how to 
solve them with an arbitrary (\cref{cor-aic}) or a given (\cref{cor:given}) initial credit.

\subsection{\kl{Multi-Dimensional Energy Parity Games}}
\AP
The \intro{multi-dimensional energy parity games}
are played
on finite \kl{multi-weighted game graphs} $(V,E,d)$ enriched with a
\intro{priority function} $\pi{:}\,V\to\+N_{>0}$; we let
$p$ %
be the number of distinct even
priorities.  Given an \intro{initial credit} $\vec c\in\+N^d$, we
denote by $\intro*\EnPgame_\vec c(V,E,d,p)$ the \kl{multi-dimensional energy parity game} where
Player~1 wins a play $v_0\xrightarrow{\vec w_1}v_1\xrightarrow{\vec
  w_2}v_2\cdots$ if it satisfies
\begin{itemize}
  \item the \intro{energy objective}: for all $i>0$, her energy level
    at step $i$ is non-negative on all components: $\vec c+\sum_{j\leq
      i}\vec w_j\geq\vec 0$, where comparisons are taken
    componentwise, and
  \item the \intro{parity objective}: the least priority $\pi(v_i)$ that
    appears infinitely often is odd;
\end{itemize}
Player~2 wins the play otherwise.  A \kl{multi-dimensional energy game}
ignores the parity condition---equivalently $\pi(v)=1$ for all $v\in V$.

\begin{example}
  Let us consider once more the graph of \cref{ex:wgg}.  Player~2
  wins the energy game with any initial credit: if Player~1
  eventually uses only the left (blue) loop, then the second component
  will eventually become negative, and similarly for the right
  (violet) loop and the first component.  Hence she must switch
  infinitely often between her two vertices using the middle (cyan)
  loop, but this decreases the 1-norm of her current energy level.
\end{example}

\begin{figure}[tbp]
  \centering
  \begin{tikzpicture}[auto,on grid]
    \node[triangle](3){$3$};
    \node[triangle, right=3.6cm of 3](2){$2$};
    \node[square,above right=.5 and 1.8 of 3](B){$2$};
    \node[square,below right=.5 and 1.8 of 3](4){$4$};
    \node[square, right=1.5 of 2](1){$1$};
    \path[->,every node/.style={font=\footnotesize,inner sep=1pt}]
      (3) edge[bend left=15,draw=black!40!cyan,pos=0.8] node{$0$} (B)
      (B) edge[bend left=15,draw=black!40!cyan,pos=0.2] node{$0$} (2)
      (2) edge[bend left=15,draw=black!40!cyan,pos=0.6] node{$0$} (4)
      (4) edge[bend left=15,draw=black!40!cyan,pos=0.35] node{$0$} (3)
      (2) edge[bend left,draw=black!40!violet] node{$-1$}(1)
      (1) edge[bend left,draw=black!40!violet] node{$0$}(2)
  ;
  \end{tikzpicture}
  \caption{\label{fig:mep}A $1$-weighted game graph with priorities.}
\end{figure}
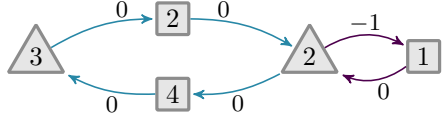
\begin{example}\label{ex:mep}
  Consider the $1$-weighted game graph with priorities
  of \cref{fig:mep}.  Player~$1$ is losing for all initial credits
  $c$ in this game: due to the energy objective the violet loop on the
  right can be played at most $c$ times, and eventually only the cyan
  loop on the left will be played, but then the parity objective is
  not satisfied by the play.
\end{example}

\subsection{\kl{Extended Multi-Dimensional Energy Games}}

\AP \kl{Extended multi-dimensional energy games} allow special weights
  (denoted by `$\omega$') that let Player~$1$ choose any value she
wants for the component.  Formally, let
$\+Z_\omega\eqdef\+Z\uplus\{\omega\}$; in infinity norms of the
extended multi-weights, $\omega$ is treated as~$1$.  An ""extended
(finite) multi-dimensional weighted game graph"" $(V,E,d)$, where
$E\subseteq (V_1\times\+Z_\omega^d\times V_2)\cup(V_2\times\+Z^d\times
V_1)$.  A play on such a graph is an infinite sequence
$v_0\xrightarrow{\vec w_1}v_1\xrightarrow{\vec w_2}v_2\cdots$ such
that $v_i\xrightarrow{\vec u_{i+1}}v_{i+1}\in E$ for all $0\leq i$ and
$\vec w_{i+1}$ instantiates $\vec u_{i+1}$ by replacing $\omega$'s
with values from $\+N$; strategies for Player~1 now have to specify
how to instantiate $\omega$'s to form plays.  Using the energy
objective as before to determine winners of plays, we obtain
the \intro{extended multi-dimensional energy game}
$\intro*\ExtEngame_\vec c(V,E,d)$ where $\vec c$ is the initial
credit~\citep{brazdil10}.

The following proposition shows how to get rid of priorities
in \kl{multi-dimensional energy parity games} at the price of extra
dimensions and the use of \kl[extended multi-dimensional energy
game]{extended games}: each even priority is associated with an extra
dimension, which is decremented by one upon entering a vertex with
this priority, and incremented by $\omega$ upon entering a vertex with
a smaller odd priority (a pair of additional vertices might need to be
introduced if the originating vertex was a Player~$2$ vertex);
see \cref{fig:eeg} for the extended multi-weighted game thus
constructed from \cref{fig:mep}.

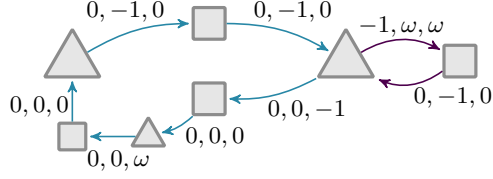
\begin{figure}[tbp]
  \centering
  \begin{tikzpicture}[auto,on grid]
    \node[triangle](3){\phantom{$3$}};
    \node[triangle, right=3.6cm of 3](2){\phantom{$2$}};
    \node[square,above right=.5 and 1.8 of 3](B){\phantom{$2$}};
    \node[square,below right=.5 and 1.8 of 3](4){\phantom{$4$}};
    \node[square, right=1.5 of 2](1){\phantom{$1$}};
    \node[square,below=1 of 3](A){};
    \node[triangle,right=1 of A](C){};
    \path[->,every node/.style={font=\footnotesize,inner sep=1pt}]
      (3) edge[bend left=15,draw=black!40!cyan,pos=0.8] node{$0,-1,0$} (B)
      (B) edge[bend left=15,draw=black!40!cyan,pos=0.2] node{$0,-1,0$} (2)
      (2) edge[bend left=15,draw=black!40!cyan,pos=0.6] node{$0,0,-1$} (4)
      (4) edge[bend left=15,draw=black!40!cyan,pos=0.35] node{$0,0,0$} (C)
      (C) edge[draw=black!40!cyan,pos=0.35] node[inner sep=5pt]{$0,0,\omega$} (A)
      (A) edge[draw=black!40!cyan,pos=0.35] node{$0,0,0$} (3)
      (2) edge[bend left,draw=black!40!violet] node{$-1,\omega,\omega$}(1)
      (1) edge[bend left,draw=black!40!violet] node{$0,-1,0$}(2)
  ;
  \end{tikzpicture}
  \caption{\label{fig:eeg}An extended $3$-weighted game graph encoding
  \cref{fig:mep}.}
\end{figure}

\begin{fact}[{\citet[Lemma~1]{jancar15}}]\label{prop-jancar}
  Let $(V,E,d)$ be a \kl{weighted game graph}, $v\in V$ an initial
  vertex, $\pi$ a \kl{priority function} with $p$ distinct even
  priorities, and $\vec c\in\+N^d$ an initial credit.  We can
  construct in logarithmic space an extended \kl{weighted game graph}
  $(V',E',d+p)$ with $V\subseteq V'$, $|V'|\leq 3|V|$,
  $|E'|\leq|E|+2|V|$, and $\|E'\|=\|E\|$ such that:
  \begin{enumerate}[(i)] 
  \item Player~1 wins $\EnPgame_\vec c(V,E,d,p)$ from $v$ if and only
  if there exists $\vec c'\in\+N^{p}$ such that she wins $\ExtEngame_{\vec c\vec c'}(V',E',d+p)$ from $v$, and
  \item Player~2 wins $\EnPgame_\vec c(V,E,d,p)$ from $v$ if and only
  if for all $\vec c'\in\+N^{p}$ he wins $\ExtEngame_{\vec c\vec c'}(V',E',d+p)$ from $v$.
  \end{enumerate}
\end{fact}

\subsection{\kl{Arbitrary Initial Credit}}
\AP
We show how \kl{extended multi-dimensional energy games},
in case the \kl{initial credit} is existentially quantified for Player~1 (this is the \intro{arbitrary initial credit problem}),
can be solved efficiently by translating them to \kl{bounding games}.
The ideas behind the translation are simple: 
enable Player~1 to keep the energy bounded at all times 
by artificial decreasing self-loops, 
and to instantiate $\omega$ weights arbitrarily 
by encoding them as increasing self-loops.
However, the proof of correctness (cf.\ \cref{pr:aic})
is unexpectedly non-trivial and makes use of \kl{perfect half space games}.

The translation is an extension of the translation in
\citep[Section~2.3]{jurdzinski15}, which did not handle $\omega$ weights,
and performs the following:
\begin{itemize}
\item 
at every vertex owned by Player~$1$ and for every coordinate~$i$,
a self-loop is inserted whose weight is the negative unit vector $-\vec{e}_i$
(these make use of new dummy Player~$2$ vertices,
to meet the requirements of player alternation and weight determinacy);
\item
for every edge whose weight~$\vec u$ is not in~$\+Z^d$,
all $\omega$~coordinates in~$\vec u$ are replaced by~$0$,
and then a dummy Player 2 vertex is inserted succeeded by 
a new Player 1 vertex that has a self-loop of weight $\vec{e}_i$ 
for each coordinate~$i$ that was $\omega$ in~$\vec u$
(the latter make use of further dummy Player~$2$ vertices as before).
\end{itemize}
\Cref{fig:bnd} illustrates this construction on the right violet loop
of the graph of \cref{fig:eeg}.
\begin{figure}[tbp]
  \centering
  \begin{tikzpicture}[auto,on grid]
    \node[triangle](L){};
    \node[triangle, right=3.6cm of L](R){};
    \node[square,above right=.5 and 1.8 of L](A){};
    \node[square,below right=.5 and 1.8 of L](B){};
    \node[square, left=1.5 of L](L1){};
    \node[square, right=1.5 of R](R1){};
    \node[square, below=1.5 of L](L2){};
    \node[square, above=1.5 of L](L3){};
    \node[square, above=1.5 of R](R2){};
    \path[->,every node/.style={font=\footnotesize,inner sep=2pt}]
      (L) edge[bend left=15,draw=black!40!violet,pos=0.8] node{$-1,0,0$} (A)
      (A) edge[bend left=15,draw=black!40!violet,pos=0.2] node{$0,0,0$} (R)
      (R) edge[bend left=15,draw=black!40!violet,pos=0.6] node{$0,0,0$} (B)
      (B) edge[bend left=15,draw=black!40!violet,pos=0.35] node{$0,0,0$} (L)
      (L) edge[bend left=15] node{$-1,0,0$} (L1)
      (L1) edge[bend left=15] node{$0,0,0$} (L)
      (R) edge[bend left=15] node{$0,1,0$}(R1)
      (R1) edge[bend left=15] node{$0,0,0$}(R)
      (R) edge[bend left=15,near end] node{$0,0,1$}(R2)
      (R2) edge[bend left=15,near start] node{$0,0,0$}(R)
      (L) edge[bend left=15,near end] node{$0,-1,0$}(L2)
      (L2) edge[bend left=15,near start] node{$0,0,0$}(L)
      (L) edge[bend left=15,near end] node{$0,0,-1$}(L3)
      (L3) edge[bend left=15,near start] node{$0,0,0$}(L)
  ;
  \end{tikzpicture}
  \caption{\label{fig:bnd}Part of the translation of \cref{fig:eeg} for
  bounding games.}
\end{figure}
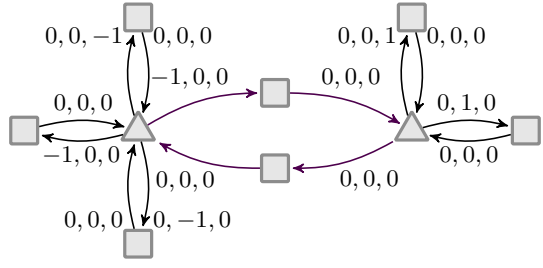

\begin{proposition}
\label{pr:aic}\intro[\symbol dagger]{}%
  Let $(V,E,d)$ be an \kl{extended multi-dimensional weighted game graph} and $v\in V$ an
  initial vertex.  We can construct in logarithmic space a \kl{multi-dimensional weighted
  game graph} $(V^\dagger,E^\dagger,d)$ with $V\subseteq V^\dagger$,
  $|V^\dagger|\leq (d+1)|V|+(d+2)|E|$, $|E^\dagger|\leq 2(d+1)|E|+2d|V|$, and
  $\|E^\dagger\|=\|E\|$ such that:
  \begin{enumerate}[(i)]
  \item\label{pr:aic.1} Player~1 wins $\ExtEngame_\vec c(V,E,d)$ 
    from $v$ for some $\vec c\in\+N^d$ if and only if she wins $\Bndgame(V^\dagger,E^\dagger,d)$ from $v$, and
  \item\label{pr:aic.2} Player~2 wins $\ExtEngame_\vec c(V,E,d)$ 
    from $v$ for all $\vec c\in\+N^d$ if and only if he wins $\Bndgame(V^\dagger,E^\dagger,d)$ from~$v$.
  \end{enumerate}
\end{proposition}

\begin{proof}
Regarding %
the right-to-left implication in item~(\ref{pr:aic.1}),
by \citep[Lemma~3.1]{jurdzinski15}, Player~1 has a winning strategy $\sigma$
in $\Bndgame(V^\dagger,E^\dagger,d)$ from $v$ that 
ensures that all total weights are at most 
$B \eqdef (4 |V^\dagger| \cdot \|E^\dagger\|)^{2 (d + 2)^3}$.
In particular, $\sigma$ does not get stuck in any of the artificial self-loops.
Hence, $\sigma$ gives rise to a winning strategy 
in $\ExtEngame_{\tuple{B, \ldots, B}}(V,E,d)$ from~$v$.

By the determinacy of bounding games 
(cf.\ \cref{th:eq.bounding}, \cref{th:tr.lmp} and \cref{th:lmp}),
it now suffices to establish 
the right-to-left implication in item~(\ref{pr:aic.2}).
Let $\tau$ be a winning strategy of Player~2 
in the perfect half space game 
$\PHSgame\tuple{\widehat{V^\dagger},\widehat{E^\dagger},d}$, 
that is positional and perfect half space oblivious,
and let $\overline{\tau}$ be its projection 
onto the extended graph $\tuple{V, E, d}$.

Consider any $\vec c\in\+N^d$, any play $\pi$ 
in $\ExtEngame_\vec c(V,E,d)$ from $v$
that is consistent with $\overline{\tau}$,
and let $\widehat{\pi^\dagger}$ be a play 
in $\PHSgame\tuple{\widehat{V^\dagger},\widehat{E^\dagger},d}$ from $v$
that corresponds to $\pi$ (i.e., where the instantiations of $\omega$ weights 
in $\pi$ are reproduced by the corresponding increasing self-loops).
Observe that any Player 2 vertex $v'$ in $\pi$
also occurs in $\widehat{\pi^\dagger}$,
and so the perfect half space chosen by $\tau$ at $v'$
must contain every negative unit vector $-\vec{e}_i$
(otherwise, Player~1 could proceed to win by repeating forever 
one of the artificial self-loops at the successor of~$v'$),
i.e., be disjoint from the non-negative orthant~$\+Q_{{\geq} 0}^d$.

Since $\tau$ is winning, there exists a partially perfect half space
$\tuple{\vec{g}_1, \ldots, \vec{g}_k}$ which is 
a prefix of all perfect half spaces that are 
chosen by $\tau$ along a suffix of $\widehat{\pi^\dagger}$,
and there exist $a_1$, $b_1$, \ldots, $a_{k - 1}$, $b_{k - 1}$ such that:
\begin{itemize}
\item 
the dot products of the total weights along $\widehat{\pi^\dagger}$
with $\vec{g}_k$ are unbounded below, and
\item
for every $\ell = 1, \ldots, k - 1$, 
the dot products of the total weights along $\widehat{\pi^\dagger}$
with $\vec{g}_\ell$ are in the interval $[a_\ell, b_\ell]$.
\end{itemize}
Hence, for the sequence of total weights along $\pi$ 
with $\vec{c}$ subtracted, the same holds.
But, by the observation above, the denotation of 
$\tuple{\vec{g}_1, \ldots, \vec{g}_k}$ 
is disjoint from the non-negative orthant, implying that
\begin{equation*}\{\vec{x} \cdot \vec{g}_k \,:\,
    \vec{c} + \vec{x} \geq \vec{0} \text{ and }\forall 1\leq\ell<k,\,
    \vec{x} \cdot \vec{g}_\ell \,\in\, [a_\ell, b_\ell]\}\end{equation*}
is bounded below.
We conclude that the total weights along $\pi$ 
are not contained in $\+N^d - \vec{c}$, 
showing as required that $\overline{\pi}$ 
is a winning strategy in $\ExtEngame_\vec c(V,E,d)$.
\end{proof}

From \cref{cor-bnd}, \cref{prop-jancar} and \cref{pr:aic}, we obtain
our first improved upper bound.

\begin{corollary}\label{cor-aic}
The \kl{arbitrary initial credit problem} for \kl{multi-dimensional
energy parity games} on
$\tuple{V, E, d}$ with
$p$ even priorities is solvable in
\ifshort\relax\else deterministic \fi time 
$(|V| \cdot \|E\|)^{O((d + p)^3 \log (d + p))}$.
\end{corollary}

We also deduce that Player~$2$ has positional winning strategies
in \kl{multi-dimensional energy parity games} with arbitrary initial
credit; this could already be derived by \cref{prop-jancar} from the
case of extended energy games with arbitrary initial credit, shown in
Lemma~19 in the \href{http://arxiv.org/abs/1002.2557}{arXiv version}
of~\citep{brazdil10}.%

\subsection{\kl{Given Initial Credit}}
\AP The \intro{given initial credit problem} for \kl{multi-dimensional
energy parity games} takes as input a \kl{multi-weighted game graph}
$\tuple{V,E,d}$, a \kl{priority function} $\pi$, an initial vertex $v$,
and an \kl{initial credit} $\vec c$ in $\+N^d$ and asks whether
Player~1 wins the \kl{multi-dimensional energy parity game}
$\EnPgame_{\vec c}(V,E,d,\pi)$ from~$v$.

Following \citep[Lemma~3.4]{jurdzinski15}, we show that any 
\kl{multi-dimensional energy
parity game} with a \kl{given initial credit} is equivalent to a \kl{bounding
game} played over a doubly-exponentially larger graph in terms of $d$,
and exponentially larger in terms of~$p$.
\begin{lemma}%
\label{lem-init}%
  \intro[\symbol ddagger]{}%
  be a \kl{multi-weighted game graph}, $\pi$ a \kl{priority function}
  with $p$ distinct even priorities, and $v\in V$.  One can construct
  in time $O(|V^\ddagger|\cdot|E|+d\cdot\log\|\vec c\|)$
  a \kl{multi-weighted game graph}
  $\tuple{V^\ddagger,E^\ddagger,d+p}$ and a vertex $v_{\vec c}$ in
  $V^\ddagger$, where $|V^\ddagger|$ is in
  $(|V|\cdot\|E\|)^{2^{O(d\log (d+p))}}$ and $\|E^\ddagger\|=\|E\|$
  such that, for all $i\in\{1,2\}$, Player~$i$ wins
  the \kl{multi-dimensional energy parity game} $\EnPgame_{\vec
  c}(V,E,d,\pi)$ from $v$ if and only if Player~$i$ wins
  the \kl{bounding game} $\Bndgame(V^\ddagger,E^\ddagger,d+p)$
  from~$v_{\vec c}$.
\end{lemma}
\begin{proof}[Proof sketch]
  We use the same arguments as in the proof
  of~\citep[Lemma~3.4]{jurdzinski15}.  The only difference is that we
  need to handle the parity condition, and thus to go through \kl{extended multi-dimensional 
  energy games} and replace \citep[Proposition~2.2]{jurdzinski15} with the
  combination of \cref{prop-jancar,pr:aic}.  These only incur a
  polynomial overhead in the size of the weighted game graphs, hence a
  bound for $|V^\ddagger|$ in $(|V|\cdot\|E\|)^{2^{O(d^\ddagger\log
  d^\ddagger)}}$ with $d^\ddagger\eqdef d+p$ can be deduced directly
  from~\citep[Lemma~3.4]{jurdzinski15}.

  We refine this bound by observing that only the first $d$ components
  of the \kl[bounding game]{$(d+p)$-dimensional bounding game} we
  construct should be treated as initialised, while the $p$
  remaining ones in \cref{prop-jancar} are arbitrary, hence the
  blowing-up construction of~\citep[Lemma~3.4]{jurdzinski15} only
  needs to be applied $d$ times, yielding instead a bound in
  $(|V|\cdot\|E\|)^{2^{O(d\log d^\ddagger)}}$; see Eq.~(9) in
  the \href{https://arxiv.org/abs/1502.06875}{arXiv version}
  of \citep{jurdzinski15}.
\end{proof}
By applying \cref{cor-bnd} to the game graph
$\tuple{V^\ddagger,E^\ddagger,d+p}$ and since $|E|\leq |V|^2$, we
obtain a \EXP[2] upper bound on the \kl{given initial credit problem},
which is again pseudo-polynomial when $d$ and $p$ are fixed.
\begin{restatable}{corollary}{corinit}
\label{cor:given}
The \kl{given initial credit problem} with \kl{initial credit} $\vec c$ for
\kl{multi-dimensional energy parity games} on $\tuple{V, E, d}$ with
$p$ even priorities is
solvable in
\ifshort\relax\else deterministic \fi time 
\begin{align*}
(|V| \cdot \|E\|)^{2^{O(d \cdot \log (d+p))}} + O(d \cdot \log \|\vec{c}\|)\ .
\end{align*}
\end{restatable}\noindent
This matches the \EXP[2] lower bound from \citep{courtois14}, and
generalises \citep[Theorem~3.5]{jurdzinski15} to \kl{multi-dimensional
energy parity games}.  Because the \kl{given initial credit problem}
for \kl{energy games} of fixed dimension $d\geq 4$ and number of even
priorities $p=0$ is already \EXP-hard~\citep{courtois14}, there is no
hope of improving the pseudo-polynomial bound in \cref{cor:given} to a
polynomial one.

\section{Concluding Remarks}
In this paper, we have shown a chain of reductions and strategy
transfers from \kl{multi-dimensional energy parity games} to \kl{perfect half
space games} and \kl{lexicographic energy games}, see \cref{fig:reds}.

There are two main outcomes.  On the complexity side, we obtain
tighter upper bounds for \kl{multi-dimensional energy parity games}, both
with \kl[arbitrary initial credit]{arbitrary} and \kl{given initial credit}.
In particular, in addition
to closing the complexity gap with \kl{given initial credit}, our \EXP[2]
upper bound in \cref{cor:given} also closes complexity gaps for
several problems already mentioned in the introduction:
\begin{itemize}
\item deciding \kl{extended multi-dimensional energy games} with
  given initial input~\citep{brazdil10},
\item deciding whether a Petri net weakly simulates a finite state
  system, or satisfies a formula of the $\mu$-calculus fragment
  defined in~\citep{abdulla13}, and
\item deciding the model-checking problem for RB$\pm$ATL~\citep{alechina16}.
\end{itemize}

The second outcome is a rather precise description of the winning
strategies for Player~$2$ in these games.  Here, the perfect half
space viewpoint is especially enlightening: Player~$2$ can win by
`announcing' in which perfect half spaces it is attempting to escape.

\appendix
\section{}%

\subsection{Proof of \cref{th:lmp}}\label{app-lmpg}
\lmpgpositional*
\begin{proof}
  We prove the lemma for Player~$2$ (in mean-payoff terminology, Min);
  the argument for Player~$1$ (Max) is analogous.  For this, let us
  fix ourselves a positional optimal strategy for Min in the
  mean-payoff game $(V, E^{(1)})$. We show that this strategy is also
  winning for Player 2 in the \kl{lexicographic energy game}
  $\LexEngame(V, E, d)$.  Hence, in the rest of the proof, we consider
  a play~$P$ consistent with this strategy in $\LexEngame(V, E, d)$,
  and we aim at showing that it is winning.
  
  \AP Let $C_1, C_2, \dots$ be the infinite sequence of simple cycles
  obtained by the `\intro{cycle decomposition}' of the play~$P$: we
  start with an empty sequence of cycles, we then push successive
  vertices of the play on a stack, and each time we push a vertex that
  is already present on the stack, we pop the resulting simple cycle
  from the top of the stack and add it to the sequence of simple
  cycles.  Observe ($\star$) that every simple cycle $C_1, C_2, \dots$
  has total multi-weight $\prec \vec 0$.  Indeed, as a cycle in the
  strategy subgraph of an optimal strategy for Min in the
  \kl{mean-payoff game} $\MPgame(V, E^{(1)})$ with a negative value,
  it has a negative total weight~\citep{zwick96}, and hence the
  observation ($\star$) follows by~\cref{prop:sign-pres}.

  \AP For a cycle~$C$ in the \kl{multi-weighted game graph} $(V, E,
  d)$, call the \intro{leading dimension} the least~$k=1,\dots, d$
  such that $\vec w(C)(k)\neq 0$ (recall that $\vec w(C)$ is the total
  \kl{multi-weight} of the edges in the cycle).  The \kl{leading
    dimension}~$k^*$ of the play~$P$ is the smallest dimension that is
  the \kl{leading dimension} of infinitely many
  cycles~$C_1,C_2,\dots$; note that in the proof for Player~$1$, $k^*$
  can equal $d+1$.

  The core of the proof is now contained in the following
  claims~\ref{cl-kstar} and~\ref{cl-k}.
  \begin{claim}\label{cl-kstar}
  For all $1 \leq\ell < k^*$,
  $\liminf_{n} \sum_{i=1}^n \vec w(C_i)(\ell) < +\infty$. 
  \end{claim}
 \begin{proof}[Proof of \cref{cl-kstar}]%
  Indeed, by the definition of~$k^*$, for all $\ell < k^*$,
  we have that $\vec w(C_i)(\ell) = 0$ for all sufficiently large~$i$.
  Hence the sequence of sums $\sum_{i=1}^n \vec w(C_i)(\ell)$ 
  is eventually constant, so its inferior limit is finite.%
  \end{proof}%

  \begin{claim}\label{cl-k} 
  If $k^*\leq d$, then $\limsup_{n} \sum_{i=1}^n \vec w(C_i)(k^*) = -\infty$.
  \end{claim}\begin{proof}[Proof of \cref{cl-k}]%
  Indeed, from the definition of~$k^*$ and the fact that
  every cycle in the decomposition has total weight $\prec \vec 0$,
  we have that $\vec w(C_i)(k^*)$ is:
  \begin{itemize}
  \item
  $\leq -1$ for infinitely many~$i$;
  \item
  $\leq 0$ for all sufficiently large~$i$. 
  \end{itemize}
  The sequence of sums $\sum_{i=1}^n \vec w(C_i)(k)$ 
  therefore has limit superior~$-\infty$.%
  \end{proof}%
  
  From the two above claims, we get that the play~$P$ is won by Min.
  The reader may worry that the expressions of the form 
  `$\sum_{i=1}^n \vec w(C_i)(\ell)$' differ from those in the
  definition of lexicographic energy games: there might be a non-empty
  simple path remaining indefinitely `on the stack' of the cycle
  decomposition, and thus not taken into account.
  However, if we want just to determine whether 
  the corresponding limit inferior (superior, resp.)
  is less than $+\infty$ (equal to $-\infty$, resp.), 
  then the discrepancy is benign because, for every simple path~$P'$, 
  we have $|\vec w(P')(k)| \leq |V| \cdot \|E\|$.
\end{proof}  

\subsection{Proof of \cref{th:tr.lmp.1}}\label{app-trlmp}
\lemtrlmp*
\begin{proof}
  Consider any infinite play
  \begin{align*}
   P&\eqdef \tuple{v_1,\vec
    H_1}\xrightarrow{\vec w_1}\tuple{v_2,\vec H_2}\xrightarrow{\vec
    w_2}\cdots
  \intertext{in the perfect half space game~$\PHSgame\tuple{\widehat{V},
    \widehat{E}, d}$ along with its corresponding play}
  \widetilde P&\eqdef
  \tuple{v_1, \vec H_1}
  \xrightarrow{\vec{e}_{\vec H_1, \vec H_2} \shuffle (\vec w_1 \cdot
    \vec H_1)}
  \tuple{v_2, \vec H_2}
  \xrightarrow{\vec{e}_{\vec H_2, \vec H_3} \shuffle (\vec w_2 \cdot
    \vec H_2)}
  \cdots
  \end{align*}
  in the lexicographic energy game
  $\LexEngame\tuple{\widehat{V}, \widetilde{E}, 2d}$.

  We show that $P$ is winning for Player~$i$ in
  $\PHSgame\tuple{\widehat{V}, \widehat{E}, d}$ if and only if $\widetilde
  P$ is winning for the same player in $\LexEngame\tuple{\widehat{V},
    \widetilde{E}, 2d}$.  This in turn entails that winning
  strategies for each player can be transferred between the two games.

  It suffices to show this for Player~$2$.  If $\widetilde P$ is winning
  for Player~$2$, then there exists $1\le k\le 2d$ such that
  $\limsup_n\sum_{j=1}^n(\vec e_{\vec H_j,\vec H_{j+1}}\shuffle\vec
  w_j \cdot \vec H_j)(k)=-\infty$ and for all $1\le\ell<k$,
  $\liminf_n\sum_{j=1}^n(\vec e_{\vec H_j,\vec H_{j+1}}\shuffle \vec
  w_j \cdot \vec H_j)(\ell)<+\infty$.  Since the coefficients $\vec
  e_{\vec H_j,\vec H_{j+1}}$ are all non-negative, $k$ cannot
  correspond to one of these dimensions. Hence $k$ is even; let $\hat
  k\eqdef k/2$.  Because $\liminf_n\sum_{j=1}^n(\vec e_{\vec H_j,\vec
    H_{j+1}}\shuffle \vec w_j \cdot \vec H_j)(\ell)<+\infty$ for all
  \emph{odd} $1\le \ell< k$, we deduce that the visited perfect half
  spaces $\vec H_1,\vec H_2,\dots$ differ on their first $\hat k$
  coordinates only finitely many times.  Hence there is an infinite
  suffix of the play where all the perfect half spaces share a common
  prefix $(\vec g_1,\dots,\vec g_{\hat k})$.  Then
  $\limsup_n\sum_{j=1}^n\vec w_j\cdot\vec g_{\hat
    k}=\limsup_n\sum_{j=1}^n(\vec e_{\vec H_j,\vec
    H_{j+1}}\shuffle\vec w_j \cdot \vec H_j)(2\hat k)=-\infty$ and for
  all $1\le\hat\ell<\hat k$, $\liminf_n\sum_{j=1}^n\vec w_j\cdot\vec
  g_{\hat\ell}=\liminf_n\sum_{j=1}^n(\vec e_{\vec H_j,\vec
    H_{j+1}}\shuffle \vec w_j \cdot \vec H_j)(2\hat\ell)<+\infty$,
  hence $P$ is also winning for Player~$2$ in the perfect half space
  game.

  Conversely, if $P$ is winning for Player~$2$ in the perfect half
  space game~$\PHSgame\tuple{\widehat{V}, \widehat{E}, d}$, then there
  is an infinite suffix starting at some index $i$ with $\vec
  G\preceq\phs_{j\geq i}\vec H_j$ satisfying \cref{eq-phs}, and let
  $k\eqdef|\vec G|$.  The $k$ first odd coordinates of the weights in
  the corresponding infinite suffix in~$\widetilde P$ are thus all
  $0$, hence the energy will not diverge on these coordinates.
  Furthermore, the $k$ first even coordinates in the same suffix are
  such that $\limsup_n\sum_{j=i}^n(\vec e_{\vec H_j,\vec
    H_{j+1}}\shuffle\vec w_j\cdot\vec H_j)(2k)=\limsup_n\sum_{j=i}^n
  \vec w_j\cdot\vec g_k=-\infty$ and, for all $1\le\ell<k$,
  $\liminf_n\sum_{j=i}^n(\vec e_{\vec H_j,\vec H_{j+1}}\shuffle\vec
  w_j\cdot\vec H_j)(2\ell)=\liminf_n\sum_{j=i}^n \vec w_j\cdot\vec
  g_\ell<+\infty$.  Thus $\widetilde P$ is also winning for Player~$2$
  in $\LexEngame\tuple{\widehat{V}, \widetilde{E}, 2d}$.
\end{proof}

\subsection{Proof of \cref{th:tr.lmp}(\ref{th:tr.lmp.3})}

\paragraph*{Key Remark}\intro[key remark]{}%
For a strategy~$\tau'$ of Player~$2$, we say that a path
is \intro{$(\tau',v)$-elementary path} if it is consistent with~$\tau'$,
it starts in some $(v,\HS)$,
ends in some $(v,\HS')$, and does not visit the vertex~$v$ in
between.  Consider
a \kl{$(\tau,v)$-elementary path} $P$ starting in~$(v,\HS)$ and ending
in $(v,\HS')$.  Then there is a \kl{$(\tau_\HS,v)$-elementary path}
$P^{\HS'}$ that is exactly like $P$ but for the fact that it begins
in~$(v,\HS')$.  This one happens to be a \emph{cycle} consistent with
$\tau_H$.  Then we clearly have
$\phspath(P)\leqpref\phspath(P^{\HS'})$,
since every \kl{perfect half space} that occurs in $P^{\HS'}$ already
occurs in $P$.  Since furthermore $\weight(P^{\HS'})=\weight(P)$, this
means that if "$P$ is winning for Player~$2$@winning play PHS", then the same holds
for~$P^{\HS'}$.

\begin{restatable}{claim}{clngphs}\label{cl-ngphs}
  If a \kl{perfect half space}~$\HS$ is \kl{bad} then
there exists a \kl{$(\tau,v)$-elementary path} starting in~$(v,\HS)$
and \kl[pwin]{losing for Player~$2$}.
\end{restatable}
\begin{proof}[Proof of \cref{cl-ngphs}]
  Indeed, if $\HS$ is bad, there exists a play resulting from playing
  a \kl{strategy} for Player~$1$ from~$(v,\HS)$ against $\tau_\HS$,
  which is winning for Player~$1$; by \cref{th:lmp,th:tr.lmp.1} we can
  assume this strategy to be positional.  Two cases may happen: Either
  this play never visits~$v$ (except at the initial position).  In
  this case, this play was already a play consistent with~$\tau$,
  contradicting the fact that the strategy~$\tau$ was winning from
  $(v,\vec H)$.  Otherwise, the infinite play encounters at least once
  more some vertex $(v,\HS')$.  Let~$P$ be the prefix of the play from
  $(v,\HS)$ to $(v,\HS')$.  This is a \kl{$(\tau,v)$-elementary path}.
  Since $P$ has been obtained from the fight of a positional strategy
  for Player~$1$ against $\tau_\HS$, the infinite play ultimately
  repeats the cycle~$P^{\HS'}$.  Thus $P^{\HS'}$ is losing for
  Player~$2$. According to the above \kl{key remark}, $P$ was thus
  already losing for Player~$2$.
\end{proof}

\clgphs*
\begin{proof}[Proof of \cref{cl-gphs}]\AP\phantomintro{proof:cl-gphs}%
Assume for the sake of contradiction that all \kl{perfect half spaces}
are \kl{bad}.  We shall prove that in this case~$\tau$ was losing
from~$(v,\HS)$.  Let us fix for all \kl{perfect half spaces}~$\HS$
a \kl{$(\tau,v)$-elementary path}~$P(\HS)$ starting from~$(v,\HS)$
ending in some $(v,f(\HS))$ and losing for Player~$2$ (it exists
according to \cref{cl-ngphs}).  Let us now construct a play consistent
with~$\tau$ starting from~$(v,\HS)$ as follows: assuming the partial
play constructed so far ends in~$(v,\HS)$, we extend it by
concatenating the path $P(\HS)$ to it, yielding a longer play ending
in $(v,f(\HS))$.  We iterate this process and, going to the limit, we
obtain an infinite play~$P$ consistent with~$\tau$.  However, this
play is an infinite concatenation of finitely many $P(\HS)$ paths,
which are all losing for Player~$2$.  Hence~$P$ is losing for
Player~$2$.  This contradicts the fact that $\tau$ was assumed to be
winning from~$(v,\HS)$. The claim is proved: there has to be
a \kl{good perfect half space}.
\end{proof}

\section*{Acknowledgements}
Work funded in part by 
the EPSRC grants EP/M011801/1 and EP/P020992/1, and
the ERC grant 259454 (GALE).
The authors thank Jo\"el Ouaknine and Prakash Panangaden 
for organising and hosting the Infinite-State Systems workshop
at Bellairs Research Institute in March 2015, where this work
started.

\bibliographystyle{abbrvnat}
\bibliography{journalsabbr,conferences,energy}

\end{document}